%% file: main.tex
\documentclass[11pt]{article}
\usepackage[margin=1in]{geometry}
\usepackage{fullpage}
\usepackage{tgtermes}
\usepackage[T1]{fontenc}
\usepackage[colorlinks,citecolor=blue,linkcolor=blue,urlcolor=blue,pagebackref]{hyperref}
\usepackage{graphicx}
\usepackage{amsfonts,amsmath,amsthm,amssymb,dsfont}
\usepackage{thm-restate}
\usepackage{xfrac,nicefrac}
\usepackage{mathdots}
\usepackage{bm,bbm}
\usepackage{url}
\usepackage{paralist}
\usepackage{enumerate}
\usepackage[normalem]{ulem}
\usepackage{xspace}
\xspaceaddexceptions{]\}}
\usepackage[capitalise]{cleveref}
\usepackage{tabu}
\usepackage{multirow}
\usepackage{framed}
\usepackage{float,wrapfig}
\usepackage{subfigure}
\usepackage{tikz}
\usetikzlibrary{decorations.pathreplacing, calligraphy}
\usepackage[usenames,dvipsnames]{pstricks}
\usepackage[linesnumbered,boxed,ruled,vlined]{algorithm2e}
\renewcommand{\epsilon}{\varepsilon}
\newcommand{\eps}{\varepsilon}

\newcommand{\poly}{\operatorname{\mathrm{poly}}}
\newcommand{\polylog}{\poly\log}

\newcommand{\R}{\mathbbm{R}}

\renewcommand{\tilde}{\widetilde}
\def\tO{\tilde{O}}

\newcommand{\vol}{\mathbf{Vol}}

\theoremstyle{plain}
\newtheorem{prob}{Problem}

\newtheorem{theorem}{Theorem}[section]
\newtheorem{lemma}[theorem]{Lemma}
\newtheorem{prop}[theorem]{Proposition}
\newtheorem{fact}[theorem]{Fact}
\newtheorem{hypo}[theorem]{Hypothesis}

\newtheorem{claim}[theorem]{Claim}
\theoremstyle{definition}
\newtheorem{definition}[theorem]{Definition}

\title{Tight Dynamic Problem Lower Bounds from \\ Generalized BMM and OMv}
\author{Ce Jin\thanks{\textcolor{black}{Supported by NSF Grant CCF-2129139. }}\\MIT\\cejin@mit.edu \and Yinzhan Xu\thanks{Supported by NSF Grant CCF-1528078.}\\MIT\\xyzhan@mit.edu}
\date{}

\begin{document}

	\setcounter{page}{0} \clearpage
	\maketitle
	\thispagestyle{empty}
	\begin{abstract}

Popular fine-grained hypotheses have been successful in proving conditional lower bounds for many dynamic problems. Two of the most widely applicable hypotheses in this context are the \emph{combinatorial Boolean Matrix Multiplication (BMM) hypothesis}  and the closely-related \emph{Online Matrix Vector Multiplication (OMv) hypothesis}. The main theme of this paper is using $k$-dimensional generalizations of these two hypotheses to prove new \emph{tight} conditional lower bounds for dynamic problems.

The \emph{combinatorial $k$-Clique hypothesis}, which is a standard hypothesis in the literature, naturally generalizes the combinatorial BMM hypothesis. In this paper, we prove tight lower bounds for several dynamic problems under the combinatorial $k$-Clique hypothesis. For instance, we show that:
\begin{itemize}
    \item  The \emph{Dynamic Range Mode} problem has no combinatorial algorithms with $\mathrm{poly}(n)$ pre-processing time, $O(n^{2/3-\epsilon})$ update time and $O(n^{2/3-\epsilon})$ query time for any $\epsilon > 0$, matching the known upper bounds for this problem. Previous lower bounds only ruled out  algorithms with $O(n^{1/2-\epsilon})$ update and query time under the OMv hypothesis.
    \item  The \emph{Dynamic Subgraph Connectivity} problem on undirected graphs with $m$ edges has no combinatorial algorithms with $\mathrm{poly}(m)$ pre-processing time, $O(m^{2/3-\epsilon})$ update time and $O(m^{1-\epsilon})$ query time for $\epsilon > 0$, matching the upper bound given by  Chan, P{\u{a}}tra{\c{s}}cu, and Roditty~[SICOMP'11], and improving the previous update time lower bound (based on OMv) with exponent $1/2$.
\end{itemize}
Other examples include tight combinatorial lower bounds for \emph{Dynamic 2D Orthogonal Range Color Counting}, \emph{Dynamic 2-Pattern Document Retrieval}, and \emph{Dynamic Range Mode} in higher dimensions.
	
Furthermore, we propose the OuMv$_k$ hypothesis as a natural generalization of the OMv hypothesis. Under this hypothesis, we prove tight lower bounds for various dynamic problems. For instance, we show that:
\begin{itemize}
    \item The \emph{Dynamic Skyline Points Counting} problem in $(2k-1)$-dimensional space has no algorithm with $\mathrm{poly}(n)$ pre-processing time and $O(n^{1-1/k-\epsilon})$ update and query time for $\epsilon > 0$, even if the updates are semi-online. 
\end{itemize}
Other examples include tight conditional lower bounds for (semi-online) Dynamic Klee's measure for unit cubes, and high-dimensional generalizations of Erickson’s problem and Langerman’s problem. 
	\end{abstract}
	\newpage
\maketitle

\section{Introduction}

\input{intro}

\section{Preliminaries}
\input{prelim}

\section{Lower Bounds under the \texorpdfstring{$k$}{k}-Clique Hypothesis}
\input{k_clique_sec_overview}

\subsection{Range Mode}
\input{mode}

\subsection{\texorpdfstring{$st$}{st} Subgraph Connectivity}
\input{stcon}

\subsection{Dynamic 2-Pattern Document Retrieval}

\input{2document}

\subsection{2D Orthogonal Range Color Counting}
\input{color}

\section{Higher Pre-Processing Time Lower Bounds}
\input{prepo}

\section{Geometric Problems and OuMv\texorpdfstring{$_k$}{k} Hypothesis}
\input{geo_sec_overview}
\input{skyline}

\section{Generalizations of OMv-Hard Problems}
\input{generalization_sec_overview}
\input{Erickson_etc}

\section*{Acknowledgements}
We would like to thank Virginia Vassilevska Williams for many helpful discussions during the early phase of this project. We also thank her for valuable comments on a draft of this paper.

\bibliographystyle{alpha} 
\bibliography{main}

\appendix
\section{Appendix}
\input{appendix}

\end{document}

%% file: intro.tex
In dynamic (data structure) problems, we need to maintain some data $D$ (e.g., graphs, sequences, geometric objects) that undergoes small updates, and to support querying $f(D)$ for some function $f$. 
Such problems are motivated by practical scenarios where we want to maintain large data sets that are constantly changing, such as social network graphs, large collaborative documents, or real-time flight trackers.
A large body of work in theoretical computer science has been devoted to designing efficient data structures to solve dynamic problems.
These data structures are not only useful on their own,
but also turn out to have applications in solving static problems in many areas of computer science, e.g., computational geometry \cite{ShamosH76}, optimization~\cite{CohenLS21}, and graph theory \cite{GoldbergT88,Cabello19}.

Some dynamic problems have efficient data structures that only require sub-polynomial  time for each update and query. One such example is the Graph Connectivity problem, where we need to maintain an undirected graph under edge insertions and deletions, and support querying whether two vertices are connected~\cite{connectivity1, connectivity3, connHolmLT01,connectivity2,connKapronKM13,connGibbKKT15,conn-Wulff-Nilsen17,chuzhoy2020deterministic}. 
However, many other dynamic problems only have way slower data structures that run in polynomial time in the data size. For instance, if we change the graph in the Graph Connectivity problem from undirected to directed (known as the Dynamic Reachability problem), the current best data structure runs in $O(n^{1.407})$ time per update or query~\cite{Sankowski04,van2019dynamic}. It is thus natural to seek lower bounds for such problems. Unfortunately, proving unconditional super poly-logarithmic  data structure lower bounds is beyond the reach of current techniques~\cite{CliffordGL15}. 

People have thus tried to prove  conditional lower bounds for dynamic problems. An important tool for proving conditional lower bounds is \emph{fine-grained complexity} (see \cite{williams2018some} for a survey),  which uses fine-grained reductions to prove conditional lower bounds for various computational problems under some hypotheses.
There has been a great success in proving dynamic problem lower bounds under various popular hypotheses, including the $3$SUM hypothesis \cite{patrascu2010towards, kopelowitz2016higher, AbboudW14, AbboudWY18, Dahlgaard16}, the APSP hypothesis \cite{RodittyZ11, AbboudW14, williams2020monochromatic, AbboudWY18, Dahlgaard16, AbboudD16, dynamiclislb}, the Strong Exponential Time Hypothesis (SETH)~\cite{AbboudW14, AnconaHRWW19, AbboudWY18, Dahlgaard16}, the combinatorial Boolean Matrix Multiplication (BMM) hypothesis~\cite{RodittyZ11, AbboudW14, clifford2018upper} and the Online Matrix Vector Multiplication (OMv) hypothesis~\cite{HenzingerKNS15, dynamiclislb, berkholz2017answering, Dahlgaard16, clifford2018upper, lau2021algorithms}. 

The  combinatorial BMM hypothesis and the closely-related OMv hypothesis are two versatile hypotheses that have been used in proving conditional lower bounds for a wide range of dynamic problems.

In the BMM problem, one is asked to compute the product of two given $n \times n$ matrices over the Boolean semiring $\{0,1\}$.
We could of course use any fast matrix multiplication algorithm to solve BMM in $O(n^{\omega})$ time, where $\omega < 2.37286$ \cite{alman2021refined} denotes the square matrix multiplication exponent. However, fast matrix multiplication algorithms use ``Strassen-like'' techniques (see e.g.~\cite{ballard2013graph}) that do not perform well in practice. This has motivated the study of ``combinatorial'' algorithms for BMM that do not use any heavy algebraic techniques, in the hope of getting a both theoretically and practically fast algorithm for BMM. Unfortunately, despite considerable amount of efforts \cite{arlazarov1970economical, bansal2009regularity, chan2014speeding, yu2018improved}, the current fastest combinatorial BMM algorithm runs in $n^3 (\log \log n)^{O(1)}/ (\log n)^4$ time \cite{yu2018improved}, still not gaining any polynomial speed-up over the brute-force $O(n^3)$ time algorithm.
Therefore, the \emph{combinatorial BMM hypothesis}, which states that no combinatorial algorithm for BMM can run in $O(n^{3-\epsilon})$ time for $\epsilon > 0$, is popular in fine-grained complexity\footnote{Throughout this work, we consider the word-RAM model of computation with $O(\log n)$-bit words.}. 

Historically, the combinatorial BMM hypothesis was first used to prove conditional lower bounds for static problems. Lee \cite{lee2002fast} first used the combinatorial BMM hypothesis to show a lower bound for Context Free Grammar Parsing. Following this work, BMM-based conditional lower bounds have found a wide range of applications, including Colored Orthogonal Range Counting \cite{KaplanRSV08}, Range Mode \cite{ChanDLMW14}, $2$-Pattern Document Retrieval \cite{larsen2015hardness},  $6$-Cycle Detection in undirected graphs \cite{DahlgaardKS17}, and many more, e.g.,  \cite{williams2018subcubic, chan2020range, DurajK0W20}. The combinatorial BMM hypothesis is also widely used in the context of dynamic problems. For instance, under the combinatorial BMM hypothesis, Roditty and Zwick  \cite{RodittyZ11} showed the hardness of partially dynamic unweighted Single-Source Shortest Paths problems,
Abboud and Vassilevska Williams~\cite{AbboudW14} showed the hardness of a variety of dynamic graph problems and some dynamic set problems, and Clifford, Gr{\o}nlund, Larsen, and Starikovskaya \cite{clifford2018upper} showed the hardness of some string problems.

Another related hypothesis that has been successful in proving dynamic problem lower bounds is the OMv hypothesis. 
In the OMv problem, we first pre-process an $n \times n$ Boolean matrix $M$, and then receive multiple length-$n$ Boolean vectors arriving one by one, and we are asked to compute the Boolean product $Mv$ for each received vector $v$ in an online fashion.
The \emph{OMv hypothesis}, proposed by Henzinger, Krinninger, Nanongkai, and Saranurak \cite{HenzingerKNS15}, states that there is no algorithm for OMv with $\poly(n)$ pre-processing time and $O(n^{1-\epsilon})$  query time for $\epsilon > 0$.\footnote{The original version of the OMv hypothesis is defined for the OMv problem with $n$ queries, but it was shown to be equivalent to the version with an arbitrary polynomial number of queries  \cite{HenzingerKNS15}.} The OMv problem can be viewed as an online version of BMM, and it was proposed in order to remove the ``combinatorial'' notion in the combinatorial BMM hypothesis. Another key advantage of OMv-based lower bounds for dynamic problems is that they hold even when the algorithms are allowed to have arbitrary polynomial pre-processing time \cite{HenzingerKNS15}. Prior to~\cite{HenzingerKNS15}, this type of results were only seen in some SETH-based lower bounds in \cite{AbboudW14}. 

People have established a wide range of hardness results based on the OMv hypothesis.  \cite{HenzingerKNS15} showed over $15$ tight hardness results under the OMv hypothesis, including many dynamic graph problems, Erickson’s problem, Pagh's problem, and the Multiphase problem. Following this work,  the OMv hypothesis has been applied  to more problems, such as database query problems~\cite{berkholz2017answering}, dynamic string problems~\cite{clifford2018upper}, $2$D range query problems \cite{lau2021algorithms}, and Dynamic Longest Increasing Subsequence \cite{dynamiclislb}.

\vspace{0.1cm}
In this paper, we study the natural high-dimensional generalizations of the BMM hypothesis and the OMv hypothesis, and show \emph{tight} conditional lower bounds for a wide range of dynamic problems under these hypotheses.

\paragraph{Combinatorial $k$-Clique hypothesis.} It is known that via combinatorial reductions, BMM is subcubically equivalent to the Triangle Detection problem, which asks to determine whether an $n$-node graph contains a triangle \cite{williams2018subcubic}. Therefore, the combinatorial BMM hypothesis is equivalent to the hypothesis stating that there is no truly subcubic combinatorial  algorithm for Triangle Detection. 

The natural generalization of Triangle Detection is $k$-Clique Detection, which asks to determine whether an $n$-node graph contains a $k$-clique, for any constant $k \ge 3$. Although the current fastest algorithm for $k$-Clique Detection runs in $O(n^{\omega(\lfloor k/3 \rfloor, \lceil k/3 \rceil, \lceil (k-1)/3 \rceil)})$ time~\cite{ItaiR78,nesetril1985,eisenbrand2004complexity}, where $\omega(a, b, c)$ denotes the exponent for  multiplying an $n^a \times n^b$ matrix with an $n^b \times n^c$ matrix, the algorithm heavily relies on fast matrix multiplication, and is thus not efficient in practice. If we restrict the algorithm to be combinatorial, then there is currently no combinatorial algorithm for $k$-Clique that runs polynomially faster than $O(n^k)$-time brute-force, for any constant $k$. 
Therefore, the following combinatorial $k$-Clique hypothesis is a popular natural generalization of the combinatorial BMM hypothesis. 

\begin{hypo}[Combinatorial $k$-Clique Hypothesis]
\label{hypo:kclique}
There is no $O(n^{k-\eps})$ time \emph{combinatorial} algorithm for $k$-Clique Detection on $n$-vertex graphs, for any $\eps >0$.
\end{hypo}

We remark that it is not new to study $k$-Clique Detection in the context of fine-grained complexity (e.g. some previous works include \cite{Chan10, bringmann2017dichotomy, BringmannW17, abboud2018valiant, Chang19, Li19,AbboudGIKPTUW19, GutenbergWW20}). For instance, Chan~\cite{Chan10} reduced $k$-Clique Detection to the $k$-dimensional Klee's measure problem, showing a matching combinatorial lower bound for the latter problem;
Bringmann, Gr{\o}nlund, and Larsen \cite{bringmann2017dichotomy} reduced $k$-Clique Detection to the Word Break problem; Abboud, Backurs, Vassilevska Williams \cite{abboud2018valiant} reduced $k$-Clique Detection to Context-Free Grammar Parsing. Nonetheless, previous applications of $k$-Clique Detection to dynamic problems are much rarer. To the best of our knowledge, the only known example is a result by Gutenberg, Vassilevska Williams, and Wein \cite{GutenbergWW20}, who reduced $4$-Clique Detection to partially dynamic Single Source Shortest Path.

\paragraph{OuMv$_k$ hypothesis. } 

A problem that is often used as an intermediate step in showing OMv-based lower bounds is the OuMv problem. In OuMv, we need to first pre-process an $n \times n$ Boolean matrix $M$. Then for each pair of length $n$ Boolean vectors $u, v$ that arrive in an online fashion, we need to compute $u^T M v$. The OuMv hypothesis states that there is no algorithm for OuMv with polynomial pre-processing time and  $O(n^{2-\eps})$  query time for $\eps > 0$. It was shown that the OMv hypothesis is equivalent to the OuMv hypothesis~\cite{HenzingerKNS15}.

We propose the following OuMv$_k$ problem for any constant integer $k \ge 2$, which is a natural high-dimensional generalization of the OuMv problem (OuMv is equivalent to OuMv$_2$).

\begin{definition}[OuMv$_k$ Problem]
During pre-processing we are given a subset $M\subseteq [n]^k$.\footnote{We use $[n]$ to denote the set $\{1, 2, \ldots, n\}$.}
Then we receive online queries each specifying $k$ sets $U^{(1)},U^{(2)},\dots,U^{(k)} \subseteq [n]$, and we need to answer whether $U^{(1)}\times U^{(2)}\times \dots\times U^{(k)}$ has a non-empty intersection with $M$.
\end{definition}

Clearly, we can handle each OuMv$_k$ query in $O(n^k)$ time, by explicitly computing $U^{(1)}\times U^{(2)}\times \dots\times U^{(k)}$ and then comparing it with $M$. We propose the following OuMv$_k$ hypothesis which states that the $O(n^k)$ brute-force algorithm is essentially the best. 

\begin{hypo}[OuMv$_k$ Hypothesis]
\label{hypo:oumvk1}
 There is no algorithm for the OuMv$_k$ problem with $n$ queries in  $O(n^{1 + k-\eps})$ total time (pre-processing time plus total query time) for any $\eps > 0$. 
\end{hypo}

Using techniques similar to \cite{HenzingerKNS15}, we can show that the following hypothesis is equivalent. For completeness, we include a proof in the appendix. 

\begin{hypo}
\label{hypo:oumvk2}
 There is no algorithm for the OuMv$_k$ problem with $\poly(n)$ pre-processing time and \sloppy $O(n^{\gamma + k-\eps})$ total query time for $n^\gamma$ queries, for any $\gamma, \eps>0$.
\end{hypo}

In the following, we explain why we believe the OuMv$_k$ Hypothesis is plausible.
The OuMv$_k$ problem can be viewed as a variant of the $(k+1)$-Clique Detection problem in $(k+1)$-partite graphs. Imagine we have $k$ vertex parts $V_1, \ldots, V_k$ each of size $n$, and we add a hyperedge among $(v_1, \ldots, v_k)$ if and only if $(v_1, \ldots, v_k) \in M$. Each OuMv$_k$ query represents a vertex $u$ in a $(k+1)$-th vertex part, and for each $i \in [k]$, we connect $u$ with $v_i \in V_i$ if and only if $v_i \in U^{(i)}$. Clearly, the answer to the OuMv$_k$ query is YES if and only if $u$ is in a ``$(k+1)$-clique'' with some vertices $v_1, \ldots, v_k$, where there is a hyperedge among $(v_1, \ldots, v_k)$ and there is an edge between $u$ and $v_i$ for every $i \in [k]$. 
Since hyperedges are more powerful than edges, and online vertices are harder than static vertices, the OuMv$_k$ problem is clearly harder than $(k+1)$-Clique Detection. 

All known algorithms \cite{ItaiR78,nesetril1985,eisenbrand2004complexity} for $(k+1)$-Clique Detection that are polynomially faster than brute-force use the following idea: grouping the vertex parts to three groups, reducing $(k+1)$-Clique Detection to Triangle Detection where each group corresponds to one vertex part in the Triangle Detection instance, and finally using fast (rectangular) matrix multiplication to solve the Triangle Detection instance. If we try to apply this idea to OuMv$_k$, we have to assign $V_1, \ldots, V_k$ to at most $2$ groups, since otherwise there is no way to encode the hyperedges. This leaves $V_{k+1}$ to its own group.

Recall the vertices in $V_{k+1}$ arrive in an online fashion. Thus, we have essentially reduced OuMv$_k$ to a Triangle Detection instance in a tripartite graph, where vertices in one vertex part arrive in an online fashion. It can be further viewed as a possibly rectangular instance of OuMv, which is known to be equivalent to OuMv \cite{HenzingerKNS15}.

Therefore, to solve OuMv$_k$ polynomially faster than the $O(n^k)$ time per query brute-force algorithm, we either need a new algorithm for $(k+1)$-Clique Detection that is drastically different from all previous algorithms, or a polynomially faster algorithm for OuMv. Thus, it is natural to consider the OuMv$_k$ hypothesis.

Gutenberg, Vassilevska Williams, and Wein \cite{GutenbergWW20} studied another generalization of the OuMv hypothesis, the OMv$3$ hypothesis, which was used to show conditional lower bound for partially dynamic Single Source Shortest Paths.
In contrast to our OuMv$_3$ problem defined on $3$-dimensional tensors, their OMv$3$ problem is defined on matrices, and admits speedup via fast matrix multiplication. Hence, their hypothesis is based on an easier problem with a lower hypothesized running time exponent, and is not directly comparable to our OuMv$_3$ hypothesis.

\subsection{Our Contributions}

\paragraph*{Tight combinatorial lower bounds based on the $k$-Clique hypothesis.} We show tight combinatorial lower bounds for dynamic problems such as Dynamic Range Mode, Dynamic Subgraph Connectivity,  and Dynamic $2$D Orthogonal Range Color Counting. These tight lower bounds are not known to be possible under either the BMM hypothesis or the OMv hypothesis. Moreover, all these lower bounds hold even if the data structures are allowed to use arbitrary
polynomial pre-processing time. Interestingly, the static variants of many problems we study had tight combinatorial lower bounds based on the BMM hypothesis, such as Range Mode \cite{ChanDLMW14} and $2$D Orthogonal Range Color Counting \cite{kaplan2007counting}; we in turn design tight combinatorial lower bounds for their dynamic variants under the combinatorial $4$-Clique hypothesis. This identifies an intriguing pattern that relates $k$-Clique-based lower bound for a static problem and $(k+1)$-Clique-based lower bound for its dynamic variant. We believe this pattern could potentially be useful for designing combinatorial clique-based lower bounds for many other dynamic problems.

We also show that many previous BMM-based lower bounds for dynamic problems \cite{AbboudW14} can be easily strengthened to arbitrary
polynomial pre-processing time under the combinatorial $4$-Clique hypothesis, without lowering the combinatorial lower bounds on update or query time.

\paragraph{Tight lower bounds based on the OuMv$_k$ hypothesis.} 
There are many problems that are parameterized by some constant integer parameters and become much harder when the integer parameters increase. For instance, such parameters could be the dimension for computational geometry problems or tensor problems, or edge cardinality for hypergraph problems. It is thus natural to seek conditional lower bounds parameterized by such integer parameters. 
However, it is unclear how to use the OMv hypothesis to explain the increased difficulties of this type of problems when the parameters increase.  Using our proposed OuMv$_k$ hypothesis, we are able to show increasing conditional lower bounds for such problems when their parameters increase. Such examples include Dynamic Skyline Points Counting, Dynamic Klee's measure for unit hypercubes, high-dimensional Erickson’s problem and high-dimensional Langerman’s problem. We believe the OuMv$_k$ hypothesis could potentially have further applications in proving dynamic problem lower bounds.

\input{table}

\subsubsection{Lower Bounds based on \texorpdfstring{$k$}{k}-Clique} 

\paragraph*{Dynamic Range Mode.}

Given an integer sequence $a_1,a_2,\dots,a_n$, a range mode query $l, r$ asks to report the integer that appears most frequently (breaking ties arbitrarily) among $a_l,a_{l+1},\dots,a_r$. 
In the Dynamic Range Mode problem, we need to maintain an integer sequence that undergoes insertions and deletions, and support range mode queries. In a certain batched version of static range mode, Batch Range Mode \cite{ChanDLMW14}, we are given a  length $n$ sequence and $n$ range mode queries, and need to answer these $n$ queries at once. 

The time complexity of combinatorial algorithms for Batch Range Mode is quite well-understood. It is known that we can solve Batch Range Mode in $\tO(n^{1.5})$ time\footnote{In this paper, $\tO$ hides poly-logarithmic factors in the input size.} by combinatorial algorithms, and any polynomially faster ($O(n^{1.5-\eps})$ time for $\eps > 0$)  combinatorial algorithm will contradict the combinatorial BMM hypothesis~\cite{ChanDLMW14}. Faster algorithms are known for Batch Range Mode if fast matrix multiplication is allowed~\cite{williams2020truly, Gu0WX21}. 

In contrast, time complexity of combinatorial algorithms for Dynamic Range Mode was much less understood. 
There are in fact multiple combinatorial algorithms for Dynamic Range Mode with $\tilde O(n^{2/3})$ update and query times \cite{ChanDLMW14, minority1}, and the $\tilde O(n^{2/3})$ bound was a seeming barrier faced by these algorithms. On the lower bound side, the combinatorial $n^{1.5-o(1)}$ lower bound of Batch Range Mode under the combinatorial BMM hypothesis  \cite{ChanDLMW14} can be easily adapted to show an $n^{0.5-o(1)}$ per update and query (non-combinatorial) lower bound under the OMv hypothesis, which has a big gap from the $\tilde O(n^{2/3})$ upper bound. People have used fast matrix multiplication to design $O(n^{2/3-\eps})$  time (for $\eps > 0$) algorithms for Dynamic Range Mode \cite{SandlundX20,Gu0WX21}, but besides the lack of progress of purely combinatorial algorithms, 
there was no other evidence why fast matrix multiplication is necessary. In fact, even an $\tO(n^{1/2})$ time combinatorial algorithm could exist under previous knowledge. 

We finally resolve this gap between the upper and lower bounds for combinatorial Dynamic Range Mode. We show that the previous $\tilde O(n^{2/3})$ seeming barrier for combinatorial Dynamic Range Mode algorithms is actually supported by a strong reason: assuming the combinatorial $4$-Clique hypothesis, no combinatorial algorithm for Dynamic Range Mode can have $\poly(n)$ pre-processing time and $O(n^{2/3-\eps})$  update and query time for any $\eps > 0$.

Our techniques can also show conditional lower bound for a similar problem, Dynamic Range Minority~\cite{minority1,minorityChanDSW15}, which asks for the least frequent integer (that appears at least once) in the query range $a_l,a_{l+1},\dots,a_r$ for a dynamic sequence $a$. 

\begin{restatable}{theorem}{modeLowerBound}
\label{thm:mode}
Assuming the combinatorial $4$-Clique hypothesis, there is no combinatorial data structure that solves Dynamic Range Mode in $\poly(n)$ pre-processing time, $O(n^{2/3-\eps})$ amortized query time and $O(n^{2/3-\eps})$ amortized update time for $\eps > 0$.  The same lower bound also holds for the Dynamic Range Minority problem.
\end{restatable}

Our techniques generalize to high-dimensional Range Mode as well, which was studied in \cite{ChanDLMW14}.

\paragraph*{Subgraph Connectivity.}

In the Subgraph Connectivity problem (SubConn) \cite{FrigioniI00,Chan06}, we need to pre-process a static undirected graph $G=(V,E)$ with $|V|=n$ and $|E|=m$, and maintain a dynamic vertex subset $S\subseteq V$ that undergoes insertions and deletions. For each query specified by vertices $s,t$, we need to report whether $s$ and $t$ are connected in the induced subgraph of $S$ in $G$.

By running breadth-first search for every query or update, it is trivial to solve SubConn in $O(m)$ query time and $O(1)$ update time, or $O(1)$ query time and $O(m)$ update time respectively. 
The first nontrivial solution was an algorithm given by Chan~\cite{Chan06} that uses fast matrix multiplication, with $\tilde O(m^{0.94})$ amortized update time and $\tilde O(m^{1/3})$ worst-case query time.
The  algorithm with current fastest update time, due to Chan, P{\u{a}}tra{\c{s}}cu, and Roditty \cite{ChanPR11}, has 
$\tilde O(m^{2/3})$ amortized update time  and $\tilde O(m^{1/3})$ worst-case query time, and does not need fast matrix multiplication.
There exist other algorithms \cite{Duan10, DuanZ17, baswana2016dynamic, chen2016improved} that achieve some combinations of almost linear space, worst-case update time guarantee, or different update-query time trade-off, but the $\tilde O(m^{2/3})$ time per update bound remains unbeaten.

The algorithm of Chan, P{\u{a}}tra{\c{s}}cu, and Roditty \cite{ChanPR11} also supports the following trade-off: for any parameter $1\le \Delta\le n$, their data structure can achieve $\tilde O(\Delta^2 + m/\Delta)$ update time  and $\tilde O(\Delta)$ query time  (with $\tilde O(m\Delta)$ pre-processing time). An interesting question is whether this trade-off curve is tight; in particular, it was asked in \cite{ChanPR11} as an open question whether the $m^{2/3}$ update time can be improved (while keeping a sublinear query time).

Previous conditional lower bounds have ruled out algorithms for SubConn with any of the following running times (for any $\eps>0$): 

\begin{enumerate}
\item   (under 3SUM \cite{AbboudW14}) $\tilde O(m^{4/3-\eps})$ pre-processing time, $O(m^{a-\eps})$ update time and $O(m^{2/3-a-\eps})$ query time, for any $a\in [1/6,1/3]$. 
\label{item1}

\item (under OMv \cite{HenzingerKNS15}) 
 polynomial pre-processing time, $O(m^{a-\eps})$ update time and $O(m^{1-a-\eps})$ query time, for any $a\in (0,1)$.
\label{item2}
 \item (under OMv \cite{HenzingerKNS15}) polynomial pre-processing time, $O(m^{1/2-\eps})$ update time and $O(m^{1-\eps})$ query time.  
\label{item3}
\end{enumerate}
Item~\ref{item1} and Item~\ref{item3} also apply to the easier $st$-Subconn problem, where each query involves two fixed vertices $s,t$ given during pre-processing. Item~\ref{item2} partly matches the trade-off curve, showing that the product of update time and query time cannot be much smaller than $m$. However, it remains open whether we can achieve $\tilde O(m^{2/3-\alpha})$ update time and $\tilde O(m^{1/3+\alpha})$ query time for some $\alpha>0$; Item~\ref{item3} only ruled out the possibility of $\alpha>1/6$.

We answer this open question,
showing that  Chan, P{\u{a}}tra{\c{s}}cu, and Roditty's combinatorial algorithm for SubConn \cite{ChanPR11} is near-optimal under the combinatorial 4-Clique hypothesis.  Our lower bound also holds for the easier $st$-SubConn problem.

\begin{restatable}{theorem}{SubConnLowerBound}
\label{thm:subconn}
Assuming the combinatorial 4-Clique hypothesis, there is no combinatorial algorithm that solves $st$-SubConn 
 in $poly(m)$ pre-processing time,  $O(m^{2/3-\epsilon})$ amortized update time, and $O(m^{1-\epsilon})$ amortized query  time for $\epsilon > 0$.
\end{restatable}
We leave it as an open problem to improve the $2/3$ exponent in the update time using fast matrix multiplication or determine it's impossible.

\paragraph*{Dynamic $2$-Pattern Document Retrieval.} 

In the $2$-Pattern Document Retrieval problem, one is given a list of strings $S_1, \ldots, S_D$ of total length $\sum_i |S_i| = n$, and needs to support the following query: given a pair of strings $(T_1, T_2)$, report/count all indices $i$ where $S_i$ contains both $T_1$ and $T_2$. The reporting variant was first considered by Muthukrishnan \cite{muthukrishnan2002efficient}, who gave a combinatorial data structure with $\tO(n^{1.5})$ pre-processing time and $O(|T_1|+|T_2|+\sqrt{n} + output)$ query time where $output$ is the output size, by using a previous algorithm due to Ferragina, Koudas, Muthukrishnan, and Srivastava \cite{ferragina2003two} for a related problem.
Larsen, Munro, Nielsen, and Thankachan \cite{larsen2015hardness} studied the counting variant of the $2$-Pattern Document Retrieval problem and noted that for the counting variant, the query time can be improved to $O(|T_1|+|T_2|+\sqrt{n})$.  They also provided a conditional lower bound for the counting variant of the  $2$-Pattern Document Retrieval problem, showing that any combinatorial algorithm  answering $O(n)$ queries requires $n^{1.5-o(1)}$ time under the combinatorial BMM hypothesis, even if the algorithm is only required to determine if the counts are zeros. 

Document Retrieval has a wide range of applications (see \cite{navarro2014spaces} for a survey) in many scenarios such as web search \cite{page1999pagerank}, bioinformatics \cite{bartsch2011genereporter}, software repositories \cite{linstead2007mining}, chemoinformatics~\cite{brown2005editorial} and symbolic music sequences \cite{typke2005survey} . For instance, in the important web search application, each document string $S_i$ can represent each website, and the patterns can represent  keywords sent to a search engine. It is also natural to formalize this application as a dynamic problem instead of a static one: websites are constantly down and up, and it makes sense for a search engine to have the option to only search for websites that are currently online. Thus, we propose the following natural dynamic variant of the $2$-Pattern Document Retrieval problem. 
\begin{restatable}[Dynamic $2$-Pattern Document Retrieval]{prob}{patterndef}
Given a list of strings $S_1, \ldots, S_D$ of total length $\sum_{i=1}^D |S_i| = n$, where each string is on or off, maintain a data structure that supports the following operations:
\begin{compactitem}
    \item Turn on or turn off a string;
    \item Given a pair of strings $(T_1, T_2)$, count the number of $i$ such that $S_i$ is on and contains both $T_1$ and $T_2$. 
\end{compactitem}
\end{restatable}

We show that this problem can be solved by a combinatorial data structure with $\tO(n^{2/3})$ time per update and $\tO(|T_1| + |T_2| + n^{2/3})$ time per query. Under the combinatorial $4$-Clique hypothesis, this data structure is in fact optimal among combinatorial ones. 

\begin{restatable}{theorem}{documentLowerBound}
\label{thm:2pattern}
Assuming the combinatorial $4$-Clique hypothesis, there is no combinatorial data structure that solves the Dynamic $2$-Pattern Document Retrieval problem in $\poly(n)$ pre-processing time, $O(n^{2/3-\eps})$ amortized query time and $O(n^{2/3-\eps})$ amortized update time for $\eps > 0$, even when all patterns have lengths $O(1)$ and the algorithm is only required to determine if the counts are zeros. 
\end{restatable}

\paragraph*{Dynamic 2D Orthogonal Range Color Counting.}

In the Orthogonal Range Color Counting problem, we are given a set of $n$ points in $\mathbb{R}^d$, each associated with a color. Each query is given as an axis-aligned box, asking the number of distinct colors of points in the box. 
The Orthogonal Range Color Counting problem and its reporting variants have been extensively studied, e.g., \cite{KaplanRSV08, color1, color2, color3, color4, color5, color6, color7, color8}. In this paper, we focus on the $2$-dimensional case. 

There are data structures with $\tO(n^2)$ pre-processing time and $\tO(1)$ query time for static 2D Orthogonal Range Color Counting \cite{color3, KaplanRSV08, munro2015range, color6}. More generally, Kaplan, Rubin, Sharir, and Verbin \cite{KaplanRSV08}  gave a data structure for static $2$D Orthogonal Range Color Counting with  a trade-off between pre-processing time and query time (see \cref{sec:color} for more details about this trade-off). 

Kaplan, Rubin, Sharir, and Verbin \cite{KaplanRSV08} also showed that, assuming the combinatorial BMM hypothesis, no algorithm can answer $n$ static 2D Orthogonal Range Color Counting queries in $O(n^{1.5-\eps})$ time for $\eps > 0$. This is tight for combinatorial algorithms as we can use their trade-off to obtain a combinatorial algorithm that solves $n$ static 2D Orthogonal Range Color Counting queries in $\tO(n^{1.5})$ time. 

As we will show in \cref{sec:color}, their trade-off also implies a combinatorial data structure with $\tO(n^{2/3})$ update and query time for the dynamic version of $2$D Orthogonal Range Color Counting where it is allowed to insert or delete points. 
Based on ideas from the reduction in \cite{KaplanRSV08}, we show that this combinatorial data structure is essentially optimal under the combinatorial $4$-Clique hypothesis. 

\begin{restatable}{theorem}{colorLowerBound}
\label{thm:color}
Assuming the combinatorial $4$-Clique hypothesis, there is no combinatorial data structure that solves Dynamic 2D Orthogonal Range Color Counting in $\poly(n)$ pre-processing time, $O(n^{2/3-\eps})$ amortized query time and $O(n^{2/3-\eps})$ amortized update time for $\eps > 0$.  
\end{restatable}

\paragraph{High pre-processing time bound for other dynamic graph problems. } 

One main weakness of the combinatorial BMM hypothesis is that it traditionally does not imply update and query lower bounds for data structures that can use arbitrary polynomial pre-processing time. Using reductions from \cite{HenzingerKNS15}, we could obtain lower bounds for data structures with arbitrary polynomial pre-processing time under the OMv hypothesis (and thus the same combinatorial lower bounds under the combinatorial BMM hypothesis), but the lower bounds for update and query times might get lower. 

For instance, under the combinatorial BMM hypothesis, Abboud and Vassilevska Williams \cite{AbboudW14} showed that no combinatorial algorithm can achieve $O(n^{3-\eps})$  pre-processing time and $O(n^{2-\eps})$ update and query times for the Dynamic $st$-Reachability problem ($st$-Reach), in which one needs to maintain a directed graph undergoing edge insertions and deletions, and needs to answer whether a fixed node $s$ can reach a fixed node $t$. Their lower bound for update and query times are very high, showing that any combinatorial algorithm essentially needs to run breadth-first search from scratch for each update or query. On the other hand, their lower bound for the pre-processing time is less desirable: their  bound does not rule out combinatorial algorithms with, say, $O(n^3)$ pre-processing time and $\tO(1)$  update and query times. 

\cite{HenzingerKNS15} improved the lower bound for pre-processing, by showing that under the OMv hypothesis, there is no algorithm for $st$-Reach that achieves $\poly(n)$ time pre-processing, $O(n^{1-\eps})$ time update and $O(n^{2-\eps})$ time query for $\eps > 0$.
Note that because $st$-Reach has fast algebraic algorithms that run in $\poly(n)$ pre-processing and $O(n^{1.407})$ time per update and query \cite{van2019dynamic}, the drop of the update time bound is inevitable for conditional lower bounds of general algorithms. 

We resolve the gap of pre-processing time for $st$-Reach, in the world of combinatorial algorithms. We show that, under the combinatorial $4$-Clique hypothesis, no combinatorial algorithm can achieve $\poly(n)$  pre-processing time and $O(n^{2-\eps})$ update and query times for $st$-Reach. We also show similar results for other dynamic graph problems such as Dynamic Strong Connectivity. See Section~\ref{sec:prepo} for the definitions of these problems.

\begin{restatable}{theorem}{polypre}
\label{thm:polypre}
Assuming the combinatorial $4$-Clique hypothesis, there is no combinatorial data structure that solves $st$-Reach, Dynamic Strong Connectivity, or Dynamic Bipartite Perfect Matching, in $\poly(n)$ pre-processing time, $O(n^{2-\eps})$ amortized query time and $O(n^{2-\eps})$ amortized update time for $\eps > 0$.  
\end{restatable}

\subsubsection{Lower Bounds based on OuMv\texorpdfstring{$_k$}{k}}

\paragraph*{Skyline Points Counting.}

Given a set $P$ of points in $\mathbb{R}^d$, a point $p \in P$ is called a \emph{skyline point} if there does not exist another point $q \in P$ such that $p_i \le q_i$ for every $i \in [d]$ (a.k.a. $q$ dominates $p$). 
In the Dynamic Skyline Points Counting problem, we need to maintain a set $P \subseteq \mathbb{R}^d$ of at most $n$ points that undergoes insertions and deletions, and query the number of skyline points in $P$.  

The Skyline Counting problem (and its variants) has been studied in various settings, e.g., \cite{KalavagattuDKS11,RahulJ12,BrodalL14,Chan20a}. For $d \le 2$, Dynamic Skyline Points Counting can be solved in amortized $\polylog(n)$ time per update \cite{OvermarsL81}. For $d = 3$, 
Chan \cite{Chan20a} designed a data structure for Skyline Points Counting in $\R^3$ with $\tilde O(n)$ pre-processing and $\tilde O(n^{2/3})$ amortized insertion and deletion time. No nontrivial upper bound is known for the fully dynamic Skyline Counting problem when $d > 3$. 

Better upper bounds are possible in the easier \emph{semi-online model},  where during the insertion of a point we are told when the point is to be deleted.  By  adapting the techniques of Chan \cite{dyn:Chan03}, it is possible to solve Skyline Points Counting in $\R^{2k-1}$ with $\tilde O(n^{1-1/k})$ time per semi-online update for any $k \ge 2$. We show that this upper bound is tight (for odd dimension) under the OuMv$_{k}$ conjecture. 

\begin{restatable}{theorem}{skyline}
\label{thm:skyline}
Let $k \ge 2$ be a positive integer. 
Assuming the OuMv$_k$ hypothesis, there is no data structure for Dynamic Skyline Points Counting in $\mathbb{R}^{2k-1}$ with $\poly(n)$ pre-processing time, $O(n^{1-1/k -\epsilon})$ amortized update and query time for $\epsilon > 0$, even in the semi-online model. 
\end{restatable}
We leave it as an open problem to determine the correct exponent for even dimensions $d\ge 4$.

Independent to our work, Dallant and Iacono \cite{dallant2021conditional} recently showed an $n^{1/2 - o(1)}$ update and query lower bound for Dynamic Skyline Points Counting in $\mathbb{R}^3$ based on the OMv hypothesis, which is the special case of Theorem~\ref{thm:skyline} for $k = 2$. They also showed an $n^{1/3-o(1)}$ lower bound for the same problem under the $3$SUM hypothesis.

\paragraph*{Klee's measure.}
The Klee's measure problem \cite{klee1977can} is an important problem in computational geometry. In the original Klee's measure problem, one is given $n$ axis-aligned boxes in $\mathbb{R}^d$, and needs to determine the volume of their union. For $d \le 2$, this problem can be solved in $O(n \log n)$ time \cite{preparata2012computational}. In higher dimensions, the best algorithms run in $\tO(n^{d/2})$ time \cite{overmars1991new, Chan10, Chan13}. It is known that no combinatorial algorithm can improve this bound significantly, under the combinatorial $k$-Clique hypothesis~\cite{Chan10}. 
As an important special case, when the input boxes are guaranteed to be unit hypercubes, the upper bound can be improved to $n^{d/3+O(1)}$ time~\cite{bringmann2010klee, Chan13}.

In this paper, we consider the Dynamic Klee's measure problem for unit hypercubes: maintain a data structure for a set of at most $n$ axis-aligned unit hypercubes in $\mathbb{R}^d$ that supports inserting a unit hypercube, deleting a unit hypercube, and querying the volume of the union of the unit hypercubes. 

In the semi-online model, where during the insertion of a hypercube we are told when the hypercube is to be deleted, Dynamic Klee's measure for unit cubes can be solved in $\tilde O(\sqrt{n})$ update time in $\R^3$ \cite[Theorem 6.1]{dyn:Chan03}. We show that the $\tilde O(\sqrt{n})$ upper bound in the semi-online model is tight under the OuMv$_k$ hypothesis.
(We remark that our lower bound is not tight for dimensions higher than 3.)

\begin{restatable}{theorem}{klee}
\label{thm:klee}
Let $k \ge 2$ be a positive integer. 
Assuming the OuMv$_k$ hypothesis, there is no data structure for Dynamic Klee's measure for unit hypercubes in $\mathbb{R}^{2k-1}$ with $\poly(n)$ pre-processing time, $O(n^{1-1/k -\epsilon})$ amortized update and query time for $\epsilon > 0$, even in the semi-online model. 
\end{restatable}

Relatedly, Dallant and Iacono \cite{dallant2021conditional}  obtained various conditional lower bounds for Dynamic Klee’s Measure Problem with Squares in $\mathbb{R}^2$ of arbitrary side lengths. This is different from our problem which requires the hypercubes to have unit side lengths. 

\paragraph*{Chan's Halfspace problem.}

The following problem appears in \cite[Section4]{dyn:Chan03}. We need to maintain  a dynamic set $H$ of hyperplanes in $\R^d$, and  a dynamic set $Q$  of points in $\R^d$. Each update operation can insert (resp.\ delete) a hyperplane to (resp.\ from) $H$ or a point to (resp.\ from) $Q$. 
Define mapping $c_H\colon \R^d \to \R$ where $c_H(q)$ is the number of hyperplanes in $H$ that contain $q$.
We need to implicitly maintain the multiset of numbers $c_H(Q) = \{c_H(q): q\in Q\}$. More precisely, Chan \cite{dyn:Chan03} originally considered outputting $\square c_H(Q)$ for any fixed operator $\square$ that is decomposable and allows computing $\square (S+j)$ from $\square S$ in constant time (where $S+j = \{i+j:i \in S\}$). For instance, $\square c_H(Q)$ can be the minimum value in $c_H(Q)$ or the sum of all values in $c_H(Q)$. 
We call this problem Chan's Halfspace problem. 

Chan \cite{dyn:Chan03} solved this problem in $\tilde O(n^{1-1/d})$ amortized time per update using $O(n)$ space. Furthermore, he used this problem as an intermediate step in obtaining faster algorithms for a wide range of computational geometry problems such as the  decision version of Dynamic Hausdorff Distance and Dynamic Bichromatic Nearest Neighbor Search. Therefore, it is important to understand the computational complexity of this problem, since any improvements to it would carry over to various other problems. 

Unfortunately, we show that the $\tilde O(n^{1-1/d})$ upper bound is essentially optimal, under the OuMv$_k$ hypothesis.

\begin{restatable}{theorem}{chan}
\label{thm:chan}
Let $k \ge 2$ be a positive integer. 
Assuming the OuMv$_k$ hypothesis, there is no data structure for Chan's Halfspace problem in $\mathbb{R}^{k}$ with $\poly(n)$ pre-processing time, $O(n^{1-1/k -\epsilon})$ amortized update and query time for $\epsilon > 0$, even if we only need to output  $\min c_H(Q)$ for each query.
\end{restatable}

\paragraph{Generalizations of Known OMv-Hard Problems. } 

A wide range of dynamic problems were shown to be hard under the OMv hypothesis \cite{HenzingerKNS15}. Among these problems, many of them have natural generalizations (e.g. graph problems generalize to hypergraph problems, matrix problems generalize to tensor problems). We show that many of these generalizations in fact have tight conditional lower bounds under the OuMv$_k$ hypothesis. 

For instance, \cite{HenzingerKNS15} showed that no algorithm for the Dynamic $s$-Triangle Detection problem, in which one needs to maintain a graph undergoing edge insertions and deletions, and needs to answer whether a fixed node $s$ is in a triangle, has $\poly(n)$ pre-processing time, $O(n^{1-\eps})$ update time and $O(n^{2-\eps})$ query time for $\eps > 0$. Its natural generalization to $k$-uniform hypergraphs, Dynamic $s$-$k$-Uniform $(k+1)$-Hyperclique,  is the following: given a $k$-uniform hypergraph that undergoes hyperedge insertions and deletions, determine whether a fixed node $s$ is in  a $k$-uniform $(k+1)$-hyperclique. We show that under the OuMv$_k$ hypothesis, no algorithm for Dynamic $k$-uniform $(k+1)$-hyperclique has $\poly(n)$ pre-processing time, $O(n^{1-\eps})$ amortized update time, and $O(n^{k-\eps})$ amortized query time for $\eps > 0$. 

We also show tight conditional lower bounds for natural generalizations of Erickson's problem and  Langerman's problem \cite{patrascu2010towards}.

\subsection{Further Related Works}

P{\u{a}}tra{\c{s}}cu \cite{patrascu2010towards} was arguably the first to systematically study fine-grained conditional lower bounds for dynamic problems. In this groundbreaking work, P{\u{a}}tra{\c{s}}cu first reduced the $3$SUM problem
to some triangle reporting problem, which is then further reduced to many dynamic problems such as Dynamic Reachability, Dynamic Shortest Paths and Subgraph Connectivity. 
This series of reductions show polynomial lower bounds for dynamic problems under the $3$SUM hypothesis.
This work was later generalized by, for instance, Abboud and Vassilevska Williams~\cite{AbboudW14}, and Kopelowitz, Pettie, and Porat~\cite{kopelowitz2016higher} to show polynomial lower bounds for more problems under the $3$SUM hypothesis. Both \cite{AbboudW14} and \cite{kopelowitz2016higher} use some variants of the triangle reporting problem as  intermediate steps in their reductions. 

The APSP hypothesis is also widely used to show conditional lower bounds for dynamic problems~\cite{RodittyZ11, AbboudW14, AbboudD16, dynamiclislb}. For instance, Abboud and Dahlgaard \cite{AbboudD16} showed hardness for Dynamic APSP and Dynamic Maximum Weight Bipartite Matching in planar graphs under the APSP hypothesis. Based on their techniques, 
Gawrychowski and Janczewski \cite{dynamiclislb} proved conditional hardness for Dynamic Longest Increasing Subsequence.
Vassilevska Williams and Xu~\cite{williams2020monochromatic} related the APSP hypothesis and the $3$SUM hypothesis in the context of dynamic problem lower bounds. In~\cite{williams2020monochromatic}, they showed that the above-mentioned variants of triangle reporting problems are actually also hard under the APSP hypothesis. Combined with previous reductions from versions of triangle reporting to many dynamic problems, e.g.,  \cite{patrascu2010towards, AbboudW14, kopelowitz2016higher}, these dynamic problems get polynomial lower bounds under the APSP hypothesis as well. 

SETH is another popular hypothesis for proving dynamic problem lower bounds. 
Under SETH, Abboud and Vassilevska Williams~\cite{AbboudW14} showed tight lower bounds for some dynamic problems such as Dynamic Strongly Connected Components Counting. \cite{AnconaHRWW19} showed hardness for Dynamic Approximate Diameter and related problems under SETH.  

Abboud, Vassilevska Williams, and Yu \cite{AbboudWY18} considered an extremely weak hypothesis, which states that at least one of the $3$SUM hypothesis, the APSP hypothesis and SETH is true. Under this hypothesis, they first showed conditional lower bound for the so-called Triangle Collection problem, and then used Triangle Collection as an intermediate step to show conditional lower bounds for many dynamic problems such as the counting version of Dynamic  Single Source Reachability. Dahlgaard \cite{Dahlgaard16} later used Triangle Collection to show conditional lower bounds for dynamic and static diameter approximating  problems.

%% file: table.tex
\begin{table}[p]
\scriptsize
    \centering
    \begin{tabular}{|c|c|c|c|c|c|c|}
    \hline
        \multirow{2}{*}{Problems} & \multicolumn{3}{c|}{Lower Bounds} & \multirow{2}{*}{Hypotheses} & \multirow{2}{*}{References} & \multirow{2}{*}{Upper Bounds}\\
    \cline{2-4}
         &  Pre-processing & Update & Query  & & & \\
    \hline
    
    Dynamic Range Mode & $\poly(n)$ & \multicolumn{2}{c|}{$n^{2/3-\epsilon}$} & $4$-Clique & Thm.~\ref{thm:mode} & \cite{ChanDLMW14, minority1}\\
    \hline
    
    Dynamic Range Minority & $\poly(n)$ & \multicolumn{2}{c|}{$n^{2/3-\epsilon}$} & $4$-Clique & Thm.~\ref{thm:mode} & \cite{minorityChanDSW15, minority1} \\
    \hline
    
    \begin{tabular}{c}
         Dynamic $d$-Dimensional  \\
         Orthogonal Range Mode
    \end{tabular} & $\poly(n)$ & \multicolumn{2}{c|}{$n^{1-1/(2d+1)-\epsilon}$} & $(2d+2)$-Clique & Thm.~\ref{thm:dmode} & Prop.~\ref{prop:dmode_upper} \\
    \hline
    
    $st$ Subgraph Connectivity & $\poly(m)$ & $m^{2/3-\epsilon}$ & $m^{1-\epsilon}$ & $4$-Clique & Thm.~\ref{thm:subconn} &
    \begin{tabular}{c}
         \cite{ChanPR11}  \\
         amortized
    \end{tabular}\\
    \hline
    
    \begin{tabular}{c}
         Dynamic $2$-Pattern   \\
         Document Retrieval 
    \end{tabular} & $\poly(n)$ & \multicolumn{2}{c|}{$n^{2/3-\epsilon}$} & $4$-Clique & Thm.~\ref{thm:2pattern} & Prop.~\ref{prop:2pattern_upper} \\
    \hline
    
    \begin{tabular}{c}
         Dynamic $2$D Orthogonal    \\
         Range Color Counting 
    \end{tabular} & $\poly(n)$ & \multicolumn{2}{c|}{$n^{2/3-\epsilon}$} & $4$-Clique & Thm.~\ref{thm:color} & Prop.~\ref{prop:color_upper} \\
    \hline
    
    Dynamic $st$-Reachability & \multirow{3}{*}[-0.3cm]{$\poly(n)$} & \multicolumn{2}{c|}{\multirow{3}{*}[-0.3cm]{$n^{2-\epsilon}$}} &  \multirow{3}{*}[-0.3cm]{$4$-Clique} & \multirow{3}{*}[-0.3cm]{Thm.~\ref{thm:polypre}} & \multirow{3}{*}[-0.3cm]{trivial} \\
    \cline{1-1}
    
    \begin{tabular}{c}
         Dynamic  \\
         Strong Connectivity 
    \end{tabular}  &  & \multicolumn{2}{c|}{}   &  &  & \\
    \cline{1-1}
    
    \begin{tabular}{c}
         Dynamic Bipartite   \\
         Perfect Matching 
    \end{tabular} &  & \multicolumn{2}{c|}{}   &  &  & \\
    \hline
    
    \begin{tabular}{c}
         Dynamic Skyline Points \\
         Counting in $\R^{2k-1}$ 
    \end{tabular} & $\poly(n)$ & \multicolumn{2}{c|}{$n^{1-1/k-\epsilon}$} & OuMv$_k$ & Thm.~\ref{thm:skyline} & \begin{tabular}{c}
    Prop.~\ref{prop:skyline_upper} \\
    semi-online
    \end{tabular}\\
    \hline
    
    \begin{tabular}{c}
         Dynamic Klee’s  
         measure \\
  for         unit hypercubes in $\R^{2k-1}$ 
    \end{tabular} & $\poly(n)$ & \multicolumn{2}{c|}{$n^{1-1/k-\epsilon}$} & OuMv$_k$ & Thm.~\ref{thm:klee} &
    \begin{tabular}{c}
    \cite{dyn:Chan03} \\
    semi-online\\
    only for $k=2$ 
    \end{tabular}\\
    \hline
    
    \begin{tabular}{c}
         Chan's Halfspace problem in $\R^{k}$
    \end{tabular} & $\poly(n)$ & \multicolumn{2}{c|}{$n^{1-1/k-\epsilon}$} & OuMv$_k$ & Thm.~\ref{thm:chan} &
    \begin{tabular}{c}
         \cite{dyn:Chan03}  \\
         amortized 
    \end{tabular}\\
    \hline
    
    \begin{tabular}{c}
         Dynamic $s$-$k$-Uniform     \\
         $(k+1)$-Hyperclique 
    \end{tabular} & $\poly(n)$ & $n^{1-\epsilon}$ & $n^{k-\epsilon}$ & OuMv$_k$ & Thm.~\ref{thm:s-k-hyperclique} & trivial \\
    \hline
    
    \begin{tabular}{c}
         $k$-Dimensional      \\
         Erickson’s problem
    \end{tabular} & $\poly(n)$ & $n^{k-1-\epsilon}$ & \hspace{1pt} $n^{k-\epsilon}$ \hspace{1pt} & OuMv$_k$ & Thm.~\ref{thm:erickson} & trivial \\
    \hline
    
    \begin{tabular}{c}
         $(k-1)$-Dimensional      \\
         Langerman’s problem
    \end{tabular} & $\poly(n)$ & \multicolumn{2}{c|}{$n^{(k-1)^2/k-\epsilon}$} & OuMv$_k$ & Thm.~\ref{thm:langerman} & Prop.~\ref{prop:Langerman_algo} \\
    \hline
    \end{tabular}
    \caption{Our lower bounds for dynamic problems. The lower bounds based on the $k$-Clique hypothesis work for combinatorial algorithms and the lower bounds based on the OuMv$_k$ hypothesis work for arbitrary algorithms. \\
    The lower bounds state that there are no algorithms achieving the stated pre-processing time, update time and query time simultaneously for $\epsilon > 0$, under the corresponding hypotheses, even if the algorithms have amortized update and query time, and the updates are semi-online. \\All our lower bounds have matching upper bounds unless otherwise stated. The upper bounds column references algorithms that run in the stated pre-processing time, update time and query time simultaneously for $\epsilon = 0$, up to poly-logarithmic factors. All algorithms work for fully dynamic inputs with worst-case time guarantees, unless otherwise stated. }
    \label{tab:my_label}
\end{table}

%% file: prelim.tex
In a graph $G=(V, E)$, we use $\mathcal{N}(v)$ to denote the set of neighbors of $v \in V$. For any subset $U \subseteq V$, we use $\mathcal{N}_U(v)$ to denote $\mathcal{N}(v) \cap U$. 

By known techniques (e.g. \cite{williams2018subcubic}), the combinatorial $k$-Clique hypothesis is equivalent to the following unbalanced version. 

\begin{hypo}[Combinatorial $k$-Clique Hypothesis, unbalanced version]
\label{hypo:unbalkclique}
Let $d_1,d_2,\dots,d_k > 0$ be constant real numbers. There is no $O(n^{d_1+d_2+\dots+d_k - \eps})$-time \emph{combinatorial} algorithm for $k$-Clique on $k$-partite graphs $(V_1\cup V_2\cup \dots \cup V_k, E)$ where $|V_i|=n^{d_i}$ for $i \in [k]$, for any $\eps>0$.
\end{hypo}

\begin{fact}
\label{fact:clique_unbalanced_eq}
For any fixed $k$ and fixed $d_1,d_2,\dots,d_k > 0$, 
\cref{hypo:kclique} is equivalent to \cref{hypo:unbalkclique}.
\end{fact}

%% file: k_clique_sec_overview.tex
In this section, we show tight combinatorial lower bounds for Dynamic Range Mode, $st$ Subgraph Connectivity, Dynamic $2$-Pattern Document Retrieval and Dynamic $2$D Orthogonal Range Color Counting under the combinatorial $4$-Clique hypothesis. 

Previously, there exist known combinatorial lower bounds under the combinatorial BMM hypothesis for the static versions of  Range Mode \cite{ChanDLMW14},  $2$-Pattern Document Retrieval~\cite{larsen2015hardness} and  $2$D Orthogonal Range Color Counting \cite{KaplanRSV08}. Based on these previous reductions, we show higher lower bounds for the dynamic variants of these problems. Intuitively, the static variants of these problems are able to simulate triangles by known reductions, and the dynamic operations are able to simulate the $4$th vertex in a $4$-clique. In Section~\ref{sec:high_dim_mode}, we show hardness for the $d$-dimensional generalizations of static and dynamic Range Mode, based on the combinatorial $(2d+1)$-Clique hypothesis and combinatorial $(2d+2)$-Clique hypothesis respectively. Our results showcase an intriguing pattern that relates $k$-Clique-based lower bound for a static problem and $(k+1)$-Clique-based lower bound for its dynamic variant. We believe this pattern could be useful for designing combinatorial clique-based lower bounds in the future.

%% file: mode.tex
We first recall the definition of Dynamic Range Mode.
\begin{prob}[Dynamic Range Mode]
Maintain a data structure for an integer array $a$ of size at most $n$, and support the following operations:
\begin{itemize}
    \item Insert or delete an integer;
    \item For each query specified by $l, r$, report the most frequent integer appearing in $a_l, a_{l+1},\ldots, a_r$, breaking ties arbitrarily. 
\end{itemize}
\end{prob}

In this section, we prove  Theorem~\ref{thm:mode}:
\modeLowerBound*

\begin{proof}
Suppose for the sake of contradiction that there is a combinatorial data structure for Dynamic Range Mode in $\poly(n)$ pre-processing time, $O(n^{2/3-\eps})$ query time and $O(n^{2/3-\eps})$ update time. Let the pre-processing time of the data structure be $O(n^t)$ for some fixed constant $t$. 

We reduce from an unbalanced instance of $4$-Clique Detection, where the $4$ vertex parts $A, B, C, D$ have sizes $n^{1/3}, n^{1/3}, n^t, n^{2/3}$ respectively. By \cref{fact:clique_unbalanced_eq}, combinatorial algorithms for such a unbalanced instance of $4$-Clique Detection requires $n^{t+4/3-o(1)}$ time under the combinatorial $4$-Clique hypothesis.

We initialize an array of size $|A||D| + |B||D|$ as follows. The array will consist of $|A| + |B|$ blocks, where each block corresponds to a permutation of $D$. For each $a \in A$, we create a permutation $P_a$ of $D$ where the neighbors of $a$ in $D$ all occur \emph{before} the non-neighbors of $a$ in $D$. Similarly, for each $b \in B$, we create a permutation $Q_b$ of $D$ where the neighbors of $b$ in $D$ all occur \emph{after} the non-neighbors of $b$ in $D$. The resulting array is the concatenation of all $P_a$ for $a \in A$, followed by all $Q_b$  for $b \in B$. This array has size $O(n)$ and thus running the pre-processing phase of the assumed data structure for Dynamic Range Mode on it takes $O(n^t)$ time. 

Then for every $c \in C$, we start a phase by performing the following operations on the data structure. First, we insert all neighbors of $c$ in $D$ into the ``middle'' of the array where the inserted elements are after all the $P_a$ but before all the $Q_b$. Then for every pair $a \in A, b \in B$, we perform a range mode query on the range that starts with the first neighbor of $a$ in $P_a$ and ends with the last neighbor of $b$ in $Q_b$. If the mode is a common neighbor of $a, b, c$ and $a, b, c$ form a triangle, then we have found a $4$-clique; otherwise, we declare that there is no $4$-clique involving $a, b, c$ and continue to the next pair of $(a, b)$. After we are done with $c$, we remove all neighbors of $c$ inserted in the phase for $c$.

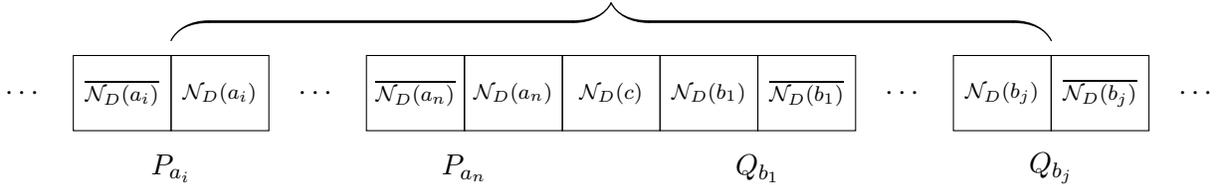
\begin{figure}[h]
\centering
\begin{tikzpicture}
	\def\D{1.3}
 	\node[] at (-\D,0) {$\cdots$};
 	
    \node[rectangle,draw, minimum width = \D cm, minimum height = 1cm] (r1) at (0*\D,0) {\scriptsize{$\overline{\mathcal{N}_D(a_i)}$}};
 	\node[rectangle,draw, minimum width = \D cm, minimum height = 1cm] (r2) at (1*\D,0) {\scriptsize{$\mathcal{N}_D(a_i)$}};
	\node[] at (\D/2,-1) {$P_{a_i}$};
	
 	\node[] at (2*\D,0) {$\cdots$};
 	
    \node[rectangle,draw, minimum width = \D cm, minimum height = 1cm] (r3) at (3*\D,0) {\scriptsize{$\overline{\mathcal{N}_D(a_n)}$}};
 	\node[rectangle,draw, minimum width = \D cm, minimum height = 1cm] (r4) at (4*\D,0) {\scriptsize{$\mathcal{N}_D(a_n)$}};
	\node[] at (3.5*\D,-1) {$P_{a_n}$};
	
	\node[rectangle,draw, minimum width = \D cm, minimum height = 1cm] (r5) at (5*\D,0) {\scriptsize{$\mathcal{N}_D(c)$}};
	
	\node[rectangle,draw, minimum width = \D cm, minimum height = 1cm] (r6) at (6*\D,0) {\scriptsize{$\mathcal{N}_D(b_1)$}};
	\node[rectangle,draw, minimum width = \D cm, minimum height = 1cm] (r7) at (7*\D,0) {\scriptsize{$\overline{\mathcal{N}_D(b_1)}$}};
	\node[] at (6.5*\D,-1) {$Q_{b_1}$};
	
	\node[] at (8*\D,0) {$\cdots$};
	
	\node[rectangle,draw, minimum width = \D cm, minimum height = 1cm] (r8) at (9*\D,0) {\scriptsize{$\mathcal{N}_D(b_j)$}};
 	\node[rectangle,draw, minimum width = \D cm, minimum height = 1cm] (r9) at (10*\D,0) {\scriptsize{$\overline{\mathcal{N}_D(b_j)}$}};
	\node[] at (9.5*\D,-1) {$Q_{b_j}$};
	
	\node[] at (11*\D,0) {$\cdots$};
	
	\draw [decorate, decoration = {calligraphic brace,raise=0.7cm, amplitude=0.5cm}, line width=1pt] (\D/2,0) --  (9.5*\D, 0);
    
\end{tikzpicture}
\caption{This figure depicts the range mode query corresponding to vertices $a_i, b_j$ and $c$. Here, $\mathcal{N}_D(v)$ denotes the set of neighbors of vertex $v$ in $D$, and $\overline{\mathcal{N}_D(v)}$ denotes the set of non-neighbors of vertex $v$ in $D$. }
\end{figure}

To show the correctness of this reduction, it suffices to show that if vertices $a, b, c$ have a common neighbor in $D$, then the range mode query corresponding to $a, b, c$ will find a common neighbor of them. This is clearly true because the range we query consists of the neighbors of $a$ in $D$, the neighbors of $b$ in $D$, the neighbors of $c$ in $D$ and some full permutations of $D$. 

The total number of updates is $O(|C||D|) = O(n^{t+2/3})$ and the total number of queries is $O(|A||B||C|) = O(n^{t+2/3})$. Therefore, the running time of the reduction is $O(n^t + n^{t+2/3} \cdot n^{2/3-\eps}) = O(n^{t+4/3-\eps})$, contradicting the combinatorial $4$-Clique hypothesis. Therefore, such an efficient combinatorial data structure for Dynamic Range Mode cannot exist under the combinatorial $4$-Clique hypothesis, leading to the claimed lower bound. 

A similar reduction works for the Dynamic Range Minority problem as well. For conciseness, we only list the main differences. First, we need to swap the order between the neighbors and non-neighbors inside each permutation $P_a$ for $a \in A$ and $Q_b$ for $b \in B$. To account for the fact that range minority queries ask for the least frequent element that needs to appear in a given range, we insert an arbitrary permutation of $D$ after all $P_a$ for $a \in A$ and before all $Q_b$ for $b \in B$, so that all elements appear in every query we make in the reduction. 
Then in phase $c$, we insert the non-neighbors of $c$ in $D$ instead of the neighbors. The query range for $a, b, c$ becomes the range  that starts with the first non-neighbor of $a$ in $P_a$ and ends with the last non-neighbor of $b$ in $Q_b$. The other parts of the reduction remain more or less the same. 
\end{proof}

\subsubsection{High-Dimensional Range Mode}
\label{sec:high_dim_mode}

A natural high-dimensional variant of Range Mode with orthogonal range queries
was studied in \cite{ChanDLMW14}, which gave a combinatorial data structure for the static version of $d$-Dimensional Range Mode with $\tO(ns^{2d-1})$ pre-processing time and $\tO(n^2/s)$ query time for any parameter $s \in [1, n]$ (their pre-processing time is implicit). By setting $s$ to be $n^{1/2d}$, their data structure implies an $\tO(n^{2-1/2d})$ time algorithm for the following Batch $d$-Dimensional  Orthogonal Range Mode problem.

\begin{prob}[Batch $d$-Dimensional  Orthogonal Range Mode]
Given $n$ points in $\R^d$ each labeled with an integer, and $n$ queries specified by $l_1,r_1,l_2,r_2,\dots,l_d,r_d$, we need to report the most frequent label appearing in the axis-aligned box $[l_1,r_1]\times [l_2,r_2]\times \dots \times [l_d,r_d]$ for each query, breaking ties arbitrarily.
\end{prob}

We show that the $\tO(n^{2-1/2d})$ time algorithm is in fact nearly-optimal under the combinatorial $(2d+1)$-Clique hypothesis.

\begin{restatable}{theorem}{StaticdmodeLowerBound}
\label{thm:static_dmode}
Assuming the $(2d+1)$-Clique hypothesis, there is no combinatorial data structure that solves Batch $d$-Dimensional  Orthogonal Range Mode in $O(n^{2-1/2d-\eps})$  time for $\eps > 0$. 
\end{restatable}

\begin{proof}
We reduce from an unbalanced instance of $(2d+1)$-Clique Detection, where the first $2d$ parts $V_1, \ldots, V_{2d}$ all have sizes $n^{1/2d}$, while the last part $V_{2d+1}$ has size $n^{1-1/2d}$. By \cref{fact:clique_unbalanced_eq}, combinatorial algorithms for such unbalanced instances of $(2d+1)$-Clique Detection require $n^{2-1/2d-o(1)}$ time under the combinatorial $(2d+1)$-Clique hypothesis.

For each $i \in [2d]$, we will first create an array $A_i$ of size $O(n)$ as follows. The elements in the array will be identified by vertices in $V_{2d+1}$. For each $v_i \in V_i$, we create a permutation of $V_{2d+1}$, such that the neighbors of $v_i$ in $V_{2d+1}$ appear before the non-neighbors of $v_i$ in $V_{2d+1}$. The array $A_i$ is then the concatenation of all the permutations. 

We can split each of the $d$ axes in $\R^d$ at the origin to get a total of $2d$ half-axes. We will put each $A_i$ on one of the half-axes as follows. For each odd $i \in [2d]$ and each $j \in [|A_i|]$, we add a point whose $\lceil i/2 \rceil$-th coordinate is $j$ and whose other coordinates are all zeros. We assign a label $A_i[j]$ to this point. For each even $i \in [2d]$ and each $j \in [|A_i|]$, we add a point whose $\lceil i/2 \rceil$-th coordinate is $-j$ and whose other coordinates are all zeros. We similarly assign a label $A_i[j]$ to this point. 

Fix a tuple $(v_1, \ldots, v_{2d}) \in V_1 \times \cdots \times V_{2d}$. For every $i \in [2d]$, we use $b_i$ to denote the index in $A_i$ of the last neighbor of $v_i$ in the permutation corresponding to $\mathcal{N}_{V_{2d+1}}(v_i)$. Then we ask a range mode query on the orthogonal range defined as the following:
\[\left\{
  \begin{array}{ll}
    x_{\lceil i/2 \rceil} \le b_i & : i \in [2d] \text{ is odd},\\
    x_{\lceil i/2 \rceil} \ge -b_i & : i \in [2d] \text{ is even}.\\
  \end{array}
\right.
\]
It is not hard to see that the multi-set of labels in this orthogonal range is exactly 
$$\{A_i[j]: 1 \le i \le 2d, 1 \le j \le b_i\}.$$
By construction of the arrays $A_i$, this multiset is the union of several full permutations of $V_{2d+1}$, and the neighborhoods of $v_1, \ldots, v_{2d}$ in $V_{2d+1}$. Thus, if $v_1, \ldots, v_{2d}$ have a common neighbor in $V_{2d+1}$, the  mode of the orthogonal range will also be a common neighbor. 

Therefore, by asking $O(n)$ range mode queries on this instance, we are able to determine whether each tuple $(v_1, \ldots, v_{2d}) \in V_1 \times \cdots \times V_{2d}$ has a common neighbor in $V_{2d+1}$, so that we can solve the $(2d+1)$-Clique Detection instance in $O(n)$ additional time. This concludes the lower bound proof for Batch $d$-Dimensional  Orthogonal Range Mode.
\end{proof}

Similar to Dynamic Range Mode, the Batch $d$-Dimensional  Orthogonal Range Mode has a natural dynamic variant. 

\begin{prob}[Dynamic $d$-Dimensional  Orthogonal Range Mode]
Maintain a data structure for a set of at most $n$ points in $\R^d$ each labeled with an integer and support the following operations:
\begin{itemize}
    \item Insert or delete a point;
    \item For each query specified by $l_1,r_1,l_2,r_2,\dots,l_d,r_d$, report the most frequent label appearing in the axis-aligned box $[l_1,r_1]\times [l_2,r_2]\times \dots \times [l_d,r_d]$, breaking ties arbitrarily. 
\end{itemize}
\end{prob}

We first show an $\tO(n^{1-1/(2d+1)})$ time per operation data structure for Dynamic $d$-Dimensional  Orthogonal Range Mode, and then show that this data structure is essentially optimal under the combinatorial $(2d+2)$-Clique hypothesis.

\begin{restatable}{prop}{dmodeUpperBound}
\label{prop:dmode_upper}
There exists a combinatorial data structure that solves  Dynamic $d$-Dimensional  Orthogonal Range Mode in $\tilde{O}(n^{2-2/(2d+1)})$  pre-processing time, $\tilde{O}(n^{1-1/(2d+1)})$  query time and $\tilde{O}(n^{1-1/(2d+1)})$  update time. 
\end{restatable}

\begin{proof}
For each label $c$, we maintain a $d$-dimensional range tree $\mathcal{T}_c$ that stores all the points with label $c$. Let $B$ be a threshold parameter to be set later. For every label $c$ that appears less than $B$ times, there are at most $O(B)$ values on each coordinate that appear as the coordinate of some point with label $c$. We consider all $O(B^{2d})$ orthogonal ranges whose boundary coordinates all appear as the coordinate of some point with label $c$. We store these orthogonal ranges in a $2d$-dimensional range tree $\mathcal{T}'$ (each boundary coordinate can be viewed as a dimension, and thus it's a $2d$-dimensional range tree), and associate it with a value equal to the number of points in the range. Note that a single range tree $\mathcal{T}'$ holds these ranges over all labels $c$. 

Then upon each update, it clearly takes $\tilde{O}(1)$ time to update the structure $\mathcal{T}_c$, and $\tilde{O}(B^{2d})$ time to update the structure $\mathcal{T}'$. 

For each query, we first enumerate all labels $c$ that appear at least $B$ times in the whole point set, and query the number of points in the queried orthogonal range via $\mathcal{T}_c$ in $\tilde{O}(1)$ time. The first potential answer is the label that appears the most times in the queried orthogonal range among these frequent labels.  
There can be at most $O(n/B)$ such labels, so it takes $\tilde{O}(n/B)$ time to handle these frequent labels. On the other hand, for labels that appear at most $B$ times, we can simply query the largest value over all the orthogonal ranges stored in $\mathcal{T'}$ that are entirely contained inside the queried orthogonal range, which takes $\tO(1)$ time. The second potential answer is the label corresponding to this largest value. 
The final answer to the query is the better one between these two potential answers.

If one of the most frequent labels $c$ appears at least $B$ times in the whole point set, then our data structure is correct since the query to $\mathcal{T}_c$ gives the correct count of label $c$ in the queried orthogonal range (and since our algorithm clearly does not over-estimate the mode). Otherwise, let $c$ be the most frequent label to the query that appears less than $B$ times in the whole point set. Without loss of generality, assume $c$ appears at least once in the queried orthogonal range (since otherwise the mode is $0$ and our algorithm must be correct since it never over-estimates the mode). Say the queried orthogonal range is $[l_1,r_1]\times [l_2,r_2]\times \dots \times [l_d,r_d]$. For each $i \in [d]$, let $l_i'$ be the smallest value that is at least $l_i$ and appears as a coordinate of some point with label $c$. Similarly, for each $i \in [d]$, let $r_i'$ be the largest value that is at most $r_i$ and appears as a coordinate of some point with label $c$. Clearly, $[l_1',r_1']\times [l_2',r_2']\times \dots \times [l_d',r_d']$ contains the same number of points with label $c$ as the queried orthogonal range does, and it is stored in the data structure $\mathcal{T}'$. Therefore, the second potential answer will be a correct range mode in this case. 

The result follows by setting $B=n^{1/(2d+1)}$. 
\end{proof}

\begin{restatable}{theorem}{dmodeLowerBound}
\label{thm:dmode}
Assuming the combinatorial $(2d+2)$-Clique hypothesis, there is no combinatorial data structure that solves  Dynamic $d$-Dimensional  Orthogonal Range Mode in $\poly(n)$ pre-processing time, $O(n^{1-1/(2d+1)-\eps})$ amortized query time and $O(n^{1-1/(2d+1)-\eps})$ amortized update time for $\eps > 0$.  
\end{restatable}

\begin{proof}
This proof combines the ideas from the proofs of \cref{thm:mode} and \cref{thm:static_dmode}. Suppose there is a combinatorial data structure for Dynamic $d$-dimensional  Orthogonal Range Mode in $\poly(n)$ pre-processing time, $O(n^{1-1/(2d+1)-\eps})$ query time and $O(n^{1-1/(2d+1)-\eps})$ update time for $\eps > 0$. Let the pre-processing time of the data structure be $O(n^t)$ for some fixed constant $t$.  

We reduce from an unbalanced instance of $(2d+2)$-Clique Detection, where the first $2d$ vertex parts $V_1, \ldots, V_{2d}$ have sizes $n^{1/(2d+1)}$. The $(2d+1)$-th vertex part $V_{2d+1}$ has size $n^t$ and the last vertex part $V_{2d+2}$ has size $n^{1-1/(2d+1)}$. 
By \cref{fact:clique_unbalanced_eq},  combinatorial algorithms for such  unbalanced instances of $(2d+2)$-Clique Detection require $n^{t+4d/(2d+1)-o(1)}$ time under the combinatorial $(2d+2)$-Clique hypothesis.

The labels of the Dynamic $d$-dimensional  Orthogonal Range Mode instance will correspond to vertices in $V_{2d+2}$. As in the proof of \cref{thm:static_dmode}, we create $O(n)$ labeled points in $d$-dimensional space that encode the neighbors of $V_1, \ldots, V_{2d}$ in $V_{2d+2}$. The key properties of the construction we need are
\begin{enumerate}
    \item We can construct these points in $O(n)$ time. 
    \item For every tuple $(v_1, \ldots, v_{2d}) \in V_1 \times \cdots \times V_{2d}$, we can find an orthogonal range in $O(1)$ time such that the multi-set of labels in the range is the union of the neighbors of $v_1, \ldots, v_{2d}$ in $V_{2d+2}$ and several copies of $V_{2d+2}$.
    \label{property2}
    \item All these orthogonal ranges contain the box $(-1, 1)^d$. \label{property3}
\end{enumerate}
Once we create these points, we use the pre-processing part of the assumed Dynamic $d$-dimensional  Orthogonal Range Mode data structure on them in $O(n^t)$ time. 

Then we start a phase for each $v_{2d+1} \in V_{2d+1}$. At the beginning of the phase, we add $|\mathcal{N}_{V_{2d+2}}(v_{2d+1}) |$ points to the data structure. The coordinates of these points can be arbitrary coordinates inside the box $(-1, 1)^d$, and the labels of these points correspond to the neighbors of $v_{2d+1}$ in $V_{2d+2}$. Then we perform several queries to the data structure. For every tuple $(v_1, \ldots, v_{2d}) \in V_1 \times \cdots \times V_{2d}$, we use the second property to find the orthogonal range corresponding to the tuple and query the mode in this range via the data structure. 
By properties~\ref{property2} and \ref{property3}, the multi-set of labels in the range will be the union of the neighbors of $v_1, \ldots, v_{2d+1}$ in $V_{2d+2}$ and several copies of $V_{2d+2}$. Thus, if $v_1, \ldots, v_{2d+1}$ have a common neighbor in $V_{2d+2}$, the data structure will return one of the common neighbors for the query. At the end of the phase, we delete all points added in this phase from the data structure. 

By previous discussion, after all phases are performed, we will know whether each tuple $(v_1, \ldots, v_{2d+1}) \in V_1 \times \cdots \times V_{2d+1}$ has a common neighbor in $V_{2d+2}$. We can then determine if the initial graph contains a $(2d+2)$-clique in $O(n^{t+2d/(2d+1)})$ time by checking if there is a tuple $v_1, \ldots, v_{2d+1}$ that forms a $(2d+1)$-clique among them and has a common neighbor in $V_{2d+2}$.  

The overall running time of the algorithm is $O(n^t+n^{t+2d/(2d+1)} + n^{t+2d/(2d+1)} \cdot n^{1-1/(2d+1)-\eps}) = O(n^{t+4d/(2d+1)-\eps})$, which contradicts the combinatorial $(2d+2)$-Clique hypothesis. 
\end{proof}

%% file: stcon.tex
Recall the definition of $st$ Subgraph Connectivity ($st$-SubConn)

\begin{prob}[$st$-SubConn]
Maintain a data structure for a static undirected graph $G=(V,E)$ with $|V|=n$ and $|E|=m$ with two fixed vertices $s, t \in V$, and a dynamic vertex subset $S\subseteq V$. Support the following operations:
\begin{itemize}
    \item Insert or delete a vertex to or from $S$;
    \item Report whether $s$ is connected to $t$ in the subgraph induced by $S$. 
\end{itemize}
\end{prob}

We will show the following lower bound for combinatorial algorithms, matching the best known upper bound~\cite{ChanPR11}. 

\SubConnLowerBound*

\begin{proof}
Suppose there is a combinatorial data structure for $st$-Subconn with $\poly(m)$ pre-processing time, $O(m^{2/3-\eps})$ update time and $O(m^{1-\eps})$ query time. Let the pre-processing time of the data structure be $O(m^r)$ for some fixed constant $r$. 

We reduce from an unbalanced instance of $4$-Clique Detection, where the $4$ vertex parts $A, B, C, D$ have sizes $m^{2/3}, m^{1/3}, m^{1/3}, m^{r}$ respectively. Let $E$ denote the edge set of this $4$-Clique Detection input instance. By \cref{fact:clique_unbalanced_eq}, any combinatorial algorithm solving such a unbalanced instance of $4$-Clique Detection requires $m^{r+4/3-o(1)}$ time under the combinatorial $4$-Clique hypothesis. 

We create an undirected graph $G$ with ``source vertex'' $s$, ``sink vertex'' $t$, and $O(m)$ edges as follows. The graph consists of disjoint vertex parts (or, ``layers'' from left to right) \[\{s\} \cup V_B \cup U_B \cup U_D \cup U_C \cup V_C \cup \{t\},\]
where $|V_B|=|B|=m^{1/3}, |V_C|=|C|=m^{1/3}, |U_B|=|U_D|=|U_C|=|A|=m^{2/3}$. We assume a natural bijection between $B$ and $V_B$, which maps $b\in B$ to $b^{V_B}\in V_B$. Similarly, $c\in C$ maps to $c^{V_C}\in V_C$, and $a\in A$ maps to $a^{U_B}\in U_B, a^{U_C}\in U_C,a^{U_D}\in U_D$. The undirected edges in $G$, defined as follows, only connect vertices between adjacent layers.
\begin{itemize}
    \item For every $b\in B$, add an edge $(s,b^{V_B})$.
    \item For every $c\in C$, add an edge $(c^{V_C},t)$.
    \item For every $b\in B,a\in A$ such that $(b,a)\in E$, add an edge $(b^{V_B},a^{U_B})$.
    \item For every $c\in C,a\in A$ such that $(c,a)\in E$, add an edge $(a^{U_C},c^{V_C})$.
    \item For every $a\in A$, add two edges $(a^{U_B},a^{U_D})$ and $(a^{U_D},a^{U_C})$.
\end{itemize}

To solve the input 4-Clique Detection instance, we use Algorithm~\ref{algo:stconn}  with the help of an $st$-SubConn data structure on $G$ maintaining an active vertex subset $S$ that undergoes insertions and deletions. 

\begin{algorithm}
Initialize the $st$-SubConn data structure on $G$, letting $S$ contain all vertices. \\
\For{$d\in D$}{
\For{$a\in A$}
{
Let $a^{U_D} \in S$ if and only if $(a,d) \in E$.\label{line:ad}
}
\For{$b\in B$ such that $(b,d)\in E$}{
Let $b^{V_B}\in S$.\label{line:b1}\\
\For{$b' \in B\setminus \{b\}$}{
Let $(b')^{V_B}\notin S$.\label{line:b3}\\}
\For{$c\in C$}{
 Let $c^{V_C}\in S$ if and only if $(c,d),(c,b)\in E$.\label{line:b2}
}
\If{$s,t$ are connected in the induced subgraph of $S$ \label{line:quer}}{
\Return{True}
}
}
}
\Return{False}
	\caption{the reduction from $4$-Clique Detection to $st$-SubConn}
\label{algo:stconn}	
\end{algorithm}

Now we prove the correctness of Algorithm~\ref{algo:stconn} solving 4-Clique Detection. First, assume the input graph contains a 4-clique with vertices $a_0,b_0,c_0,d_0$. Then, at Line~\ref{line:quer} when $d=d_0,b=b_0$, we must have $a_0^{U_D}\in S$ (due to Line~\ref{line:ad}), $b_0^{V_B}\in S$ (due to Line~\ref{line:b1}), and $c_0^{V_C} \in S$ (due to Line~\ref{line:b2}). From the definition of $G$, we also know that $G$ contains edges $(b_0^{V_B},a_0^{U_B}), (a_0^{U_C},c_0^{V_C})$ since $(b,a),(a,c)\in E$. Hence, there is a path $s\to b_0^{V_B}\to a_0^{U_B} \to a_0^{U_D} \to a_0^{U_C} \to c_0^{V_C} \to t$ in the induced subgraph of $S$, and the query at Line~\ref{line:quer} will return True.

Conversely, suppose the query at Line~\ref{line:quer} returns True. We will show that it implies the existence of a 4-clique. Let $p$ be the shortest path from $s$ to $t$ in the induced subgraph of $S$. Then, the shortest path $p$ must visit $V_B$ at most once, since otherwise we could take the last visit $b^{V_B}_{last}\in V_B$ and directly go from $s$ to $b^{V_B}_{last}$ along the edge $(s,b^{V_B}_{last})$. Also, any path from $s$ to $t$ must use at least one vertex in $V_B$, since the removal of $V_B$ disconnects $s$ and $t$.  Hence, $p$ visits exactly one vertex $b^{V_B}$ in $V_B$. By a similar argument, $p$ visits exactly one vertex $c^{V_C}$ in $V_C$. 
Then, inspecting the structure of the middle layers $U_B,U_D,U_C$, we see that between $b^{V_B}$ and $c^{V_C}$ the path $p$ must visit $a^{U_B}\to a^{U_D}\to a^{U_C}$ for some $a\in A$.

From Lines~\ref{line:b1}-\ref{line:b3} we observe that $b^{V_B}$ is the only vertex in $S\cap V_B$, and $(b,d)\in E$. Then, from Line~\ref{line:b2} and $c^{V_C}\in S$ we know $(c,d),(c,b)\in E$. From Line~\ref{line:ad} and $a^{U_D}\in S$ we know $(a,d)\in E$. Finally, from the definition of $G$, we know $(b,a),(a,c)\in E$ from the existence of edges $(b^{V_B},a^{U_B}),(a^{U_C},c^{V_C})$. So $a,b,c,d$ form a 4-clique.

It remains to analyze the time complexity of Algorithm~\ref{algo:stconn}. Line~\ref{line:ad} contributes $O(|D|\cdot |A|) = O(m^{2/3+r})$ update operations in total. Lines~\ref{line:b1}-\ref{line:b2} contribute $O(|D|\cdot |B|\cdot (|B|+|C|)) = O(m^{2/3+r})$ update operations in total.  Line~\ref{line:quer} contributes $O(|D|\cdot |B|) = O(m^{1/3+r})$ query operations in total. Hence, the overall running time of this algorithm is asymptotically at most
\[ m^r + m^{2/3+r}\cdot m^{2/3-\eps} + m^{1/3+r}\cdot m^{1-\eps} \le m^{4/3+r-\eps},\]
contradicting the combinatorial 4-Clique hypothesis.
\end{proof}

%% file: 2document.tex
Recall the definition of Dynamic $2$-Pattern Document Retrieval:
\patterndef*

In this section, we will show a combinatorial data structure for Dynamic $2$-Pattern Document Retrieval with $\tilde{O}(|T_1| + |T_2| + n^{2/3})$ query time and $\tilde{O}(n^{2/3})$ update time. Then we will show that these running times are essentially optimal under the combinatorial $4$-Clique hypothesis.

\begin{prop}
\label{prop:2pattern_upper}
There is a combinatorial data structure for the Dynamic $2$-Pattern Document Retrieval problem in $\poly(n)$ pre-processing time, $\tilde{O}(|T_1| + |T_2| + n^{2/3})$ query time and $\tilde{O}(n^{2/3})$ update time.
\end{prop}
\begin{proof}
In \cite{ferragina2003two}, 
Ferragina, Koudas, Muthukrishnan, and Srivastava reduced the  static $2$-Pattern Document Retrieval problem to the  Common Colors Query problem via a combinatorial reduction. Given the input strings, their reduction can produce in $\tO(n)$ time an array $A$ of size $O(n)$ whose elements are identified as colors  (and these colors correspond to the indices of input strings). For every query $(T_1, T_2)$, their reduction can produce two contiguous intervals of the array $A$ in $\tO(|T_1| + |T_2|)$ time. Then the set of input strings that contain both $T_1$ and $T_2$ has a bijection to the set of unique colors these two  intervals both contain.

Their reduction is also applicable to Dynamic $2$-Pattern Document Retrieval.
The colors of the array correspond to the input strings of the $2$-Pattern Document Retrieval instance, so turning on/off a string corresponds to turning on/off a color, and each query of the Dynamic $2$-Pattern Document Retrieval problem asks the number of colors that are turned on and the two given intervals share.

Based on their reduction, we construct the following data structure for Dynamic $2$-Pattern Document Retrieval. Given the input strings $S_1, \ldots, S_D$, we create the array $A$ in the Common Colors Query problem via their reduction in $\tO(n)$ time. We then aim to maintain the following sub-data structures during pre-processing or after each update:
\begin{itemize}
    \item For every color $c$ that appears at most $n^{1/3}$ times and is turned on, let $i_1, \ldots, i_k$ be all indices in $A$ that have color $c$, in increasing order. We also additionally set $i_{k+1} = |A|+1$ for notational convenience. Then for every pair $j_1, j_2 \in [k]$, we store a quadruple of integers $(i_{j_1}, i_{j_1+1}, i_{j_2}, i_{j_2+1})$ in a $4$D range tree $\mathcal{T}$. 
    
    It clearly only takes $\poly(n)$ time to create $\mathcal{T}$ during pre-processing. For each update, we turn on or off at most one color, and each color that appears at most $n^{1/3}$ times needs to store $O(n^{2/3})$ quadruples in $\mathcal{T}$, so it takes $\tO(n^{2/3})$ time to maintain $\mathcal{T}$ after each update. 
    \item For every color $c$, we also maintain a balanced search tree $\mathcal{B}_c$ that contains all the indices in the array with color $c$. 
    
    Clearly, it takes $\poly(n)$ time to create all $\mathcal{B}_c$ during pre-processing, and they don't need to be updated after each update. 
\end{itemize}

Now we describe how to handle a query given these sub-data structures. 
For each query with strings $T_1, T_2$, we first use the reduction in \cite{ferragina2003two} to compute two contiguous intervals $I_1 = [l_1, r_1]$ and $I_2=[l_2, r_2]$ in $\tO(|T_1| + |T_2|)$ time for the Common Colors Query problem. By the correctness of their reduction, it remains to determine the number of unique  colors these two intervals both contain. 
Then we  query the $4$D range tree $\mathcal{T}$ to count the number of quadruples in the box $[l_1, r_1] \times [r_1 + 1, \infty] \times [l_2, r_2] \times [r_2 + 1, \infty]$. It is not hard to see that this count equals the number of unique colors that appear at most $n^{1/3}$ times, is turned on, and are contained by both $I_1$ and $I_2$. It remains to consider colors that appear more than $n^{1/3}$ times. 
We iterate over every color $c$ that appears more than $n^{1/3}$ times and is turned on, and check whether $I_1$ and $I_2$ both contain color $c$. There are at most $O(n^{2/3})$ such colors, and each color $c$ can be checked in $\tilde{O}(1)$ time using $\mathcal{B}_c$. 

Overall, our data structure has $\tO(n^{2/3})$ update time and $\tO(|T_1| + |T_2| + n^{2/3})$ query time. 
\end{proof}

\documentLowerBound*

\begin{proof}
Suppose there is a combinatorial data structure for Dynamic $2$-Pattern Document Retrieval in $\poly(n)$ pre-processing time, $O(n^{2/3-\eps})$ query time and $O(n^{2/3-\eps})$ update time. Let the pre-processing time of the data structure be $O(n^t)$ for some fixed constant $t$. 

We reduce from an unbalanced instance of $4$-Clique Detection, where the $4$ vertex parts $A, B, C, D$ have sizes $n^{1/3}, n^{1/3}, n^t, n^{2/3}$ respectively. By \cref{fact:clique_unbalanced_eq}, combinatorial algorithms for such a unbalanced instance of $4$-Clique Detection requires $n^{t+4/3-o(1)}$ time, under the  combinatorial $4$-Clique hypothesis.

The number of strings in the Dynamic $2$-Pattern Document Retrieval instance will be $n^{2/3}$, corresponding to the vertices in $D$, and the alphabet of this instance will have size $2n^{1/3}$, with symbols corresponding to the vertices in $A \cup B$. 

For every $d \in D$, we create a string $S_d$ of length $|\mathcal{N}_{A\cup B}(d)|$ that contains one symbol for each neighbor of $d$ in $A \cup B$. The orders of these symbols can be arbitrary. Clearly, the total length of these strings is $O(|D|\cdot (|A|+|B|))\le O(n)$, so we can use the assumed data structure for Dynamic $2$-Pattern Document Retrieval to pre-process these strings in $O(n^t)$ time. 

We perform a phase for each $c \in C$. At the beginning of each phase, we perform $O(n^{2/3})$ updates on the data structure so that for every $d\in D$, $S_d$ is turned on if and only if $d$ is a neighbor of $c$. Then for every pair $(a, b) \in A \times B$, we query the data structure to determine whether there is any $S_d$ that is turned on and contains both symbols $a$ and $b$. Clearly, such a string exists if and only if $a, b, c$ have a common neighbor. 

Once we determine whether $a, b, c$ have a common neighbor for each $(a, b, c) \in A \times B \times C$, we can easily determine if the $4$-Clique Detection instance has a $4$-clique in $O(|A||B||C|)=O(n^{t+2/3})$ time. 

In total, we perform $O(n^{t+2/3})$ updates and queries on the data structure, so the overall running time of the reduction is $O(n^t + n^{t+2/3} \cdot n^{2/3-\eps}) = O(n^{t+4/3 -\eps})$. This contradicts the combinatorial $4$-Clique hypothesis, and thus the assumed data structure cannot exist under the combinatorial $4$-Clique hypothesis.
\end{proof}

%% file: color.tex
\label{sec:color}

In this section, we show an algorithm for the Dynamic $2$D Orthogonal Range Color Counting problem and a matching combinatorial conditional lower bound. 

\begin{prob}[Dynamic $2$D Orthogonal Range Color Counting]
Maintain a set of at most $n$ colored points on the $2$D plane, and support the following operations:
\begin{itemize}
    \item Insert or delete a point;
    \item Given $(x_1,x_2,y_1,y_2)$, output the number of distinct colors appearing in the rectangle $[x_1,x_2]\times [y_1,y_2]$.
\end{itemize}
\end{prob}

\begin{prop}
\label{prop:color_upper}
     There exists a combinatorial data structure for Dynamic $2$D Orthogonal Range Color Counting with $\poly(n)$ pre-processing time and $\tO(n^{2/3})$ update and query time. 
\end{prop}
\begin{proof}
Kaplan, Rubin, Sharir, and Verbin \cite{KaplanRSV08}  gave a data structure for the static version of $2$D Orthogonal Range Color Counting with  a trade-off between pre-processing time and query time. For any trade-off parameter $1 \le X \le n$, they gave a data structure with $\tO(X)$ query time and 
$$
  \begin{cases} 
   \tO\left( \frac{n^{(\omega+1)/2}}{X^{(\omega - 1) / 2}} \right) & \text{if } X \geq n^{\frac{\omega-1}{\omega+1}}, \\
   \tO\left( \frac{n^{\frac{2-\alpha \beta + 2 \beta}{\beta + 1}}}{X^{\frac{2-\alpha \beta}{\beta + 1}}}\right)       & \text{if } n^{\frac{\alpha / 2}{\alpha / 2 + 1}}  \le X < n^{\frac{\omega-1}{\omega+1}}, \\
   \tO\left( n^2/X^2 \right) & \text{if } n^{\frac{\alpha / 2}{\alpha / 2 + 1}}  > X
  \end{cases}
$$
pre-processing time, where $\alpha \ge 0.31389$ \cite{gall2018improved} is defined as $\sup \{t \ge 0: \omega(1, t, 1) = 2\}$, and $\beta$ is defined as $\frac{\omega-2}{1-\alpha}$. When restricted to combinatorial algorithms, currently $\omega = 3, \alpha = 0$ and $\beta = 1$. Therefore, their pre-processing time when restricted to combinatorial algorithms becomes $\tO(n^2/X)$ no matter what $X$ is. 

We maintain the following sub-data structures:
\begin{enumerate}
    \item For each color $c$, we maintain a $2$D range tree $\mathcal{T}_c$, storing all the coordinates of points with color $c$. It can be updated in $\tO(1)$ time per update. 
    \item After every $\tO(n^{2/3})$ updates, we rebuild the static data structure $\mathcal{D}$ from \cite{KaplanRSV08} with trade-off parameter $X =\tO(n^{2/3})$. We also build $\mathcal{T}_c^{\text{old}}$, which are copies of $\mathcal{T}_c$ at the time when we rebuild $\mathcal{D}$. It takes $\tO(n^{2/3})$ amortized time per update to maintain $\mathcal{D}$ and $\mathcal{T}_c^{\text{old}}$. 
\end{enumerate}

For each $2$D Orthogonal Range Color Counting query, we first feed the query to data structure $\mathcal{D}$ and get an outdated count. Then we enumerate all colors $c$ of points that are inserted or deleted after we last build $\mathcal{D}$, and check whether the orthogonal range contains a point in $\mathcal{T}_c$ and $\mathcal{T}_c^{\text{old}}$ respectively. If the result is different for $\mathcal{T}_c$ and $\mathcal{T}_c^{\text{old}}$, we update the count accordingly. 

Thus, we have a data structure with $\tO(n^{2/3})$ amortized update time and $\tO(n^{2/3})$ query time. The update time can be easily made to be worst-case by applying the Overmars's \textit{global rebuilding} technique \cite{overmars1983design}. 
\end{proof}

Recall Theorem~\ref{thm:color}:
\colorLowerBound*

\begin{proof}
Suppose there is a combinatorial data structure for Dynamic $2$D Orthogonal  Range Color Counting in $\poly(n)$ pre-processing time, $O(n^{2/3-\eps})$ query time and $O(n^{2/3-\eps})$ update time. Let the pre-processing time of the data structure be $O(n^t)$ for some fixed constant $t$. 

We reduce from an unbalanced instance of $4$-Clique Detection, where the $4$ vertex parts $A, B, C, D$ have sizes $n^{1/3}, n^{1/3}, n^t, n^{2/3}$ respectively. By \cref{fact:clique_unbalanced_eq},  combinatorial algorithms for such an unbalanced instance of $4$-Clique requires $n^{t+4/3-o(1)}$ time under 
the  combinatorial $4$-Clique Detection hypothesis.

In our reduction we will create an instance where  points with different colors could share the same 2D coordinate; this could be easily avoided by adding small perturbations to the coordinates, which we omit here for simplicity.

For notational convenience, we identify the vertex set $A$ with the integer set $[|A|]$, and similarly identify the vertex set $B$ (or $C$,$D$) with the integer set $[|B|]$ (or $[|C|],[|D|]$).
We initialize a set of $|A||D|+|B||D|$ colored 2D points as follows. For every $a\in A, d\in D$ such that $a$ and $d$ are adjacent, we create a point with coordinate $(a,|A|+1-a)$ and color $d$. Similarly, for every $b\in B, d\in D$ such that $b$ and $d$ are adjacent, we create a point with coordinate $(-b,-|B|-1+b)$ and color $d$.   We add $O(n)$ points in total and thus building the assumed data structure for Dynamic $2$D Orthogonal  Range Color Counting on these points takes $O(n^t)$ time. 

Then for every $c \in C$, we start a phase by performing the following operations on the data structure. 
First, for every $a \in A, b \in B$, let $q_{ab}$ denote the answer of querying the rectangle $x_1=-b,y_1=-|B|-1+b,x_2=a,y_2=|A|+1-a$.
Then, for every $d\in D$ that is adjacent to $c\in C$, we add a point with coordinate $(0,0)$ and color $d$.  Then for every pair $a \in A, b \in B$, let $q_{abc}$ denote the answer of querying (again) the rectangle $x_1=-b,y_1=-|B|-1+b,x_2=a,y_2=|A|+1-a$, and let $q_{ac}$ denote the answer of querying  the rectangle $x_1=0,y_1=0,x_2=a,y_2=|A|+1-a$, and $q_{bc}$ denote the answer of querying  the rectangle $x_1=-b,y_1=-|B|-1+b,x_2=0,y_2=0$. Observe that, by construction, $q_{abc}$ equals the number of vertices in $D$ that are adjacent to at least one of $a,b,c$, $q_{ab}$ equals the number of vertices in $D$ that are adjacent to at least one of $a, b$, and similarly for $q_{ac}$ and $q_{bc}$. Let $q_{a}$ (and $q_b,q_c$) denote the number of neighbors of $a$ (and $b,c$) in $D$. Then by the inclusion-exclusion principle, the number of vertices in $D$ that are simultaneously adjacent to $a,b,c$ equals $q_{abc}-q_{ab}-q_{bc}-q_{ac}+q_a+q_b+q_c$. We return YES if this number is non-zero and $a,b,c$ form a triangle. 
After we are done with $c$, we remove all the points added at the beginning of the phase for $c$. If we have not returned YES after we finish all the phases for all $c$, we return NO. 

The correctness of the reduction is immediate since we essentially determined whether each triple $(a, b, c)$ has a common neighbor in $D$. 

The total number of updates is $O(|C||D|) = O(n^{t+2/3})$ and the total number of queries is $O(|A||B||C|) = O(n^{t+2/3})$. Therefore, the running time of the reduction is $O(n^t + n^{t+2/3} \cdot n^{2/3-\eps}) = O(n^{t+4/3-\eps})$, contradicting the combinatorial $4$-Clique hypothesis. Therefore, such an efficient combinatorial data structure for Dynamic $2$D Orthogonal  Range Color Counting cannot exist under the combinatorial $4$-Clique hypothesis, leading to the claimed lower bound. 
\end{proof}

%% file: prepo.tex
\label{sec:prepo}

In this section, we improve the previous combinatorial lower bounds for $st$-Reach, Dynamic Strong Connectivity, and Dynamic Bipartite Perfect Matching, by showing higher pre-processing lower bounds. 

\begin{prob}[$st$-Reach]
Given a directed graph $G=(V, E)$ with $n$ vertices and two fixed nodes $s,t \in V$, we need to support edge insertions and edge deletions, and querying whether $t$ is reachable from $s$. 
\end{prob}

\begin{prob}[Dynamic Strong Connectivity]
Given a directed graph $G=(V, E)$ with $n$ vertices, we need to support edge insertions and edge deletions, and querying whether the graph is strongly connected. 
\end{prob}

\begin{prob}[Dynamic Bipartite Perfect Matching]
Given a bipartite graph $G=(V, E)$ with $n$ vertices, we need to support edge insertions and edge deletions, and querying whether the graph has a perfect matching. 
\end{prob}

Recall \cref{thm:polypre}:
\polypre*

\begin{proof}
We will only prove the statement for the $st$-Reach problem. The statements for Dynamic Strong Connectivity and Dynamic Bipartite Perfect Matching immediately follow via the reductions from $st$-Reach to Dynamic Strong Connectivity and Dynamic Bipartite Perfect Matching in  \cite{AbboudW14}.

Suppose there is a combinatorial data structure for $st$-Reach with $\poly(n)$ pre-processing time, $O(n^{2-\eps})$ query time and $O(n^{2-\eps})$ update time.
Let the pre-processing time of the data structure be $O(n^r)$ for some fixed constant $r$. 

We reduce from an unbalanced instance of $4$-Clique Detection, where the $4$ vertex parts $A, B, C, D$ have sizes $n, n, n, n^{r}$ respectively.
By \cref{fact:clique_unbalanced_eq},  combinatorial algorithms for such an unbalanced instance of $4$-Clique Detection requires $n^{r+3-o(1)}$ time under the combinatorial $4$-Clique hypothesis.

To solve this 4-Clique Detection instance, we create an $st$-Reach instance on a directed graph with eight layers of vertices, from left to right: \[\{s\},A_1,B_1,B_2,C_1,C_2,A_2,\{t\},\] where $s,t$ are the fixed source node and sink node respectively, and $A_1,A_2$ (resp.\ $B_1,B_2$ and $C_1,C_2$) are copies of the vertex set $A$ (resp.\ $B$ and $C$) of the 4-Clique Detection instance.
The edges in this directed graph will only connect adjacent layers from left to right. Between $B_2$ and $C_1$, we copy the edges between $B$ and $C$ in the 4-Clique Detection instance. Between $A_1$ and $B_1$, we copy the edges between $A$ and $B$ in the 4-Clique Detection instance. Between $C_2$ and $A_2$, we copy the edges between $C$ and $A$ in the 4-Clique Detection instance. We use the pre-processing stage of the $st$-Reach data structure to pre-process this directed graph in $O(n^r)$ time.

We iterate over all $d\in D$ and do the following for each $d$. For every $b\in B$, we use the insertion/deletion operation of the $st$-Reach data structure to connect an edge from $b\in B_1$ to $b\in B_2$ if and only if $d$ is adjacent to $b$. Similarly, we connect an edge from $c\in C_1$ to $c\in C_2$ if and only if $d$ is adjacent to $c$.  Then, for every $a\in A$ that is adjacent to $d$, we do the following: add an edge from $s$ to $a\in A_1$ and an edge from $a\in A_2$ to $t$, ask whether $t$ is reachable from $s$, and then remove the two edges just added.  Observe that there is a path $s\to a \to b\to b\to c\to c\to a\to t$ if and only if $(a,b,c,d)$ forms a 4-clique in the $4$-Clique Detection instance.

The above reduction performs one pre-processing step, $O(|D|\cdot n)$ edge updates and $O(|D|\cdot n)$ queries. Hence, we can solve the 4-Clique Detection instance in $O(n^{r} + n^{r}\cdot n\cdot n^{2-\eps})$ time, contradicting the $n^{r+3-o(1)}$ lower bound.
\end{proof}

%% file: geo_sec_overview.tex
In this section, we will show OuMv$_k$-based conditional lower bounds for Dynamic Skyline Points Counting, Dynamic Klee's measure for unit hypercubes, and Chan’s Halfspace problem. 
For certain low-dimensional cases of these problems, our lower bounds are actually based on the OMv hypothesis (i.e., OuMv$_2$). We remark that even these OMv-based lower bounds for the low-dimension problems were  not known previously in the literature. 

Starting from this section, all conditional lower bounds hold for all algorithms (not necessarily combinatorial algorithms). 

%% file: skyline.tex
\subsection{Skyline Points Counting}

In this section, we study the Dynamic Skyline Points Counting problem. We first give its formal definition. 

\begin{definition}
Given a set of points $P$ in $\mathbb{R}^d$, a point $p \in P$ is called a \emph{skyline point} or \emph{maximal point} if there does not exist another point $q \in P\setminus \{p\}$ such that $p_i \le q_i$ for every $i \in [d]$ (a.k.a. $q$ dominates $p$). 
\end{definition}

\begin{prob}[Dynamic Skyline Points Counting]
For a constant integer parameter $d \ge 1$, maintain a data structure for a set of at most $n$ points in $\mathbb{R}^d$ and support inserting a point, deleting a point, and querying the number of skyline (maximal) points. 
\end{prob}

We first show an upper bound for Dynamic Skyline Points Counting in $\R^{2k-1}$ in the semi-online model. We need the following two lemmas.

\begin{lemma}[{\cite[Lemma 2.1]{dyn:Chan03}}]
Consider a problem $\Pi$ with the following property, where $\alpha \ge 1$ and $0 < \beta \le 1$ are constants: there exists a data structure that can pre-process a set $S$ of $n$ points in $\tilde O(n)$ time, such that given any additional set $S'$ of $b$ points, the data structure can solve $\Pi$ on the set $ S \cup S'$
(block query) in $\tilde O(b^\alpha n^{1-\beta})$ time.

Then, we can solve $\Pi$ on a set of $n$ points under semi-online updates in $\tilde O(n^{1-\beta/(1+\alpha)} )$ time per update.
\label{lem:semimi}
\end{lemma}
\begin{lemma}[{\cite[Theorem 2.1]{KaplanRSV08}}]
\label{lem:decompo}
Let $A$ be a set of $n$ points in $\R^{2k-1}$. For $a = (x_1,x_2,\dots,x_{2k-1})\in A$, let $Q(a)$ denote the orthant $(-\infty,x_1]\times (-\infty,x_2]\times \cdots \times (-\infty,x_{2k-1}]$. 

We can decompose $\bigcup_{a\in A} Q(a)$ into $O(n^{k-1})$ pairwise disjoint boxes in $\tilde O(n^{k-1} )$ time.
\end{lemma}

\begin{prop}
\label{prop:skyline_upper}
For any $k \ge 2$, there exists a data structure for Dynamic Skyline Points Counting in $\R^{2k-1}$  in the semi-online model with $\poly(n)$ pre-processing time and $\tO(n^{1-1/k})$ update and query time. 
\end{prop}
\begin{proof}
We verify that the Skyline Points Counting problem satisfies the property required by \cref{lem:semimi} with $\beta = 1$ and $\alpha = k-1$, which would directly imply the statement.

Given a set $S$ of $n$ points in $\R^{2k-1}$, we first remove all the points that are dominated by some other points, and let $S_0$ denote the remaining points (i.e., $S_0$ contains all the skyline points of $S$). To check whether a point is dominated, we can use standard $(2k-1)$-dimensional range trees in $\polylog(n)$ time per query after an $\tilde O(n)$ time pre-processing. Hence, $S_0$ can be constructed in $\tilde O(n)$ time.

Then, given any additional set $S'$ of $b$ points, we solve the Skyline Point Counting problem on $S\cup S'$ as follows. First, remove all points in $S'$ that are dominated by some other points in $S'\cup S$, and let $S'_0$ denote the remaining points. The set $S'_0$ can be similarly computed as before, in $\tilde O(b)$ time. Then, observe that the Skyline Points of $S'\cup S$ consist of
\begin{compactitem}
    \item The points in $S'_0$.
    \item Points in $S_0$ that are not in $\bigcup_{p\in S'_0} Q(p)$.
\end{compactitem}
We use \cref{lem:decompo} to decompose $\bigcup_{p\in S'_0} Q(p)$ into $\tilde O(b^{k-1})$ disjoint boxes  in $\tilde O(b^{k-1})$ time. For each of the boxes, we count the number of points in $S_0$ it contains, using the $(2k-1)$-dimensional range tree. Hence, we can count the total number of skyline points of $S'\cup S$ in $\tilde O(b^{k-1})$ time.
\end{proof}

Then we show that the upper bound is nearly-optimal in the semi-online model:
\skyline*

\begin{proof}
Assume for the sake of contradiction that such an efficient data structure exists. 
We will reduce from an OuMv$_k$ instance of dimension $N = \Theta(n^{1/k})$. Let $M \subseteq [N]^k$ be the input set of OuMv$_k$. 

Let $\delta = o(1/ N)$ be a sufficiently small positive real number. 
For every tuple $(a_1, \ldots, a_k) \in M$, we create a point $$(a_1 - \delta a_k, N-a_1, a_2, N-a_2, \ldots, a_{k-1}, N - a_{k-1}, a_k) \in \R^{2k-1}$$ for the Dynamic Skyline Points Counting instance. We call these points \textit{initial points}. Then we use the pre-processing part of the assumed data structure to pre-process these points in $\poly(n)$ time. 

We first show these initial points do not dominate each other. 
\begin{claim}
Two distinct initial points do not dominate each other. 
\end{claim}
\begin{proof}
Suppose $(a_1 - \delta a_k, N-a_1, a_2, N-a_2, \ldots, a_{k-1}, N - a_{k-1}, a_k)$ is dominated by $(b_1 - \delta b_k, N-b_1, b_2, N-b_2, \ldots, b_{k-1}, N - b_{k-1}, b_k)$. For each integer $i \in [2, k-1]$, if we consider the $(2i-1)$-th and $(2i)$-th coordinates, we must have $a_i \le b_i$ and $N-a_i \le N-b_i$, which lead to $a_i = b_i$. Then we consider the first, second and last coordinates. We have $a_1 - \delta a_k \le b_1 - \delta b_k, N - a_1 \le N - b_1$ and $a_k \le b_k$. Since $a_1, a_k, b_1, b_k$ are all integers from $[N]$ and $\delta  \ll 1/N$, these inequalities imply $a_1 = b_1$ and $a_k = b_k$. 

Thus, two points can dominate each other only if they are the same. 
\end{proof}

For every OuMv$_k$ query $U^{(1)} \times \cdots \times U^{(k)}$, we perform the following phase. For every $i \in [k-1]$, and $j'  \in [N] \setminus U^{(i)}$, we insert the following point to the data structure:
$$(\underbrace{\infty, \ldots, \infty}_{2i-2 \text{ coordinates of } \infty}, j', N-j', \infty, \ldots,\infty),$$
i.e., it is a point where the $(2i-1)$-th and $(2i)$-th coordinates are $j'$ and $N-j'$ respectively, and all other coordinates are $\infty$. Then, for each $j \in \{0, \ldots, N\}$, we perform an insertion, a query, and a deletion, as follows: first insert a point 
$$(\infty, \ldots, \infty, j), $$
i.e., it is a point where the last coordinate is $j$ and all other coordinates are $\infty$. After this insertion, we query the data structure for the number of maximal points, and denote the answer of the query by $c_j$. After the query, we delete the point $(\infty, \ldots, \infty, j)$ and proceed to the next $j$. After we finish for all $j \in \{0, \ldots, N\}$, we delete all points added in the current phase and end the phase. 

Then we show that given $c_1, \ldots, c_N$, we can determine whether $U^{(1)} \times \cdots \times U^{(k)}$ intersects $M$ in $O(N)$ additional time. We first show the following claim:
\begin{claim}
\label{cl:skyline}
For any $j \in [N]$, $$c_j = - \left (\sum_{1 \le i \le k-1}|U_i| \right ) + (k-1)N  + 1 + \left| M \cap \left(U^{(1)} \times \dots \times U^{(k-1)} \times [j + 1, N] \right)\right|.$$
\end{claim}
\begin{proof}
First, it is easy to verify that the points added within each phase are always maximal points, which contributes $\sum_{1 \le i \le k-1}(N-|U_i|)   + 1$ to $c_j$. 

We then analyze which of the initial points are maximal points. Since initial points do not dominate each other,  it suffices to consider how the points added within each phase dominate the initial points. 

For every $i \in [k-1]$, and $j' \not \in U^{(i)}$, we inserted a point where the $(2i-1)$-th and $(2i)$-th coordinates are $j'$ and $N-j'$ respectively, and all other coordinates are $\infty$. By our construction of initial points, these points precisely dominate those initial points whose $(2i-1)$-th and $(2i)$-th coordinates are $j'$ and $N-j'$ respectively. These points in turn correspond to tuples in $M$ whose $i$-th entries equal $j'$. In the query for $c_j$, we also inserted another point $(\infty, \ldots, \infty, j)$, which precisely dominates those initial points whose last coordinate is at most $j$. These points in turn correspond to tuples in $M$ whose last entries are at most $j$. Therefore, the set of undominated initial points has a one-to-one correspondence with $M \cap \left(U^{(1)} \times \dots \times U^{(k-1)} \times [j+1, N] \right)$. This gives the final term of $c_j$ in the claim statement. 
\end{proof}

By Claim~\ref{cl:skyline}, for every $j \in [N]$, $c_{j-1} - c_{j} = \left| M \cap \left(U^{(1)} \times \dots \times U^{(k-1)} \times \{j\} \right)\right|$. Therefore, 
\begin{equation*}
    \begin{split}
        \left| M \cap \left(U^{(1)} \times \dots \times U^{(k-1)} \times U^k \right)\right| &= \sum_{j \in U^{(k)}}\left| M \cap \left(U^{(1)} \times \dots \times U^{(k-1)} \times \{j\} \right)\right|\\
        &= \sum_{j \in U^{(k)}} (c_{j-1} - c_{j}),
    \end{split}
\end{equation*}
which can be computed in $O(N)$ additional time. This concludes the correctness proof of the reduction. 

The pre-processing time of the reduction is clearly $\poly(n) = \poly(N)$. For each OuMv$_k$ query, we spend $O(N)$ data structure updates and queries, which take $O(N \cdot n^{1-1/k-\epsilon})=O(N^{k-k\epsilon})$ time. We also spend $O(N)$ additional time, so the running time for each query is $O(N+N^{k-k\epsilon})$. This clearly contradicts the OuMv$_k$ hypothesis. Therefore, assuming the OuMv$_k$ hypothesis, there is no data structure for Dynamic Skyline Points Counting in $\mathbb{R}^{2k-1}$ with $\poly(n)$ pre-processing time, $O(n^{1-1/k -\epsilon})$ update and query time for $\epsilon > 0$. 

Clearly, the lower bound also works for data structures in the semi-online model, since in fact, we know the deletion time of  all points when they are inserted. 
\end{proof}

\subsection{Klee's Measure for Unit Hypercubes}

We first formally define Dynamic Klee's measure for unit hypercubes.  

\begin{prob}[Dynamic Klee's measure for unit hypercubes]
For a constant integer parameter $d \ge 1$, maintain a data structure for a set of at most $n$ axis-parallel unit hypercubes in $\mathbb{R}^d$ and support inserting a unit hypercube, deleting a unit hypercube, and querying the volume of the union of the unit hypercubes. 
\end{prob}

Recall our lower bound for this problem:
\klee*
\begin{proof}
The main idea of the proof is similar to the proof of \cref{thm:skyline}. Assume for the sake of contradiction that such an efficient data structure exists. 
We will reduce from an OuMv$_k$ instance over $[N]^k$ for $N = \Theta(n^{1/k})$. Let $M \subseteq [N]^k$ be the input set of the OuMv$_k$ instance. 

Without loss of generality, we assume the side lengths of the hypercubes are $N$ by scaling up every dimension by a factor of $N$. In the proof, a hypercube with largest corner $(p_1, \ldots, p_{2k-1})$ is the hypercube $$[p_1 - N, p_1] \times \cdots \times [p_{2k-1} - N, p_{2k-1}].$$

For every point $p \in \{0, N\}^{2k-1} \setminus \{N\}^{2k-1}$, we add a hypercube with largest corner $p$. The union of these hypercubes covers all space in $[-N, N]^{2k-1}$ except the nonnegative orthant $[0, N]^{2k-1}$.  
Let $\delta = o(1/N)$ be a sufficiently small real number. For every tuple $(a_1, \ldots, a_k) \in M$, we add a hypercube with largest corner $$(a_1 - \delta a_k, N-a_1, a_2, N-a_2, \ldots, a_{k-1}, N - a_{k-1}, a_k).$$
We call all hypercubes considered so far \textit{initial hypercubes}. We then use the pre-processing part of the assumed  data structure for Dynamic Klee's measure for unit hypercubes to pre-process the initial hypercubes in $\poly(n)$ time. 

For every OuMv$_k$ query $U^{(1)} \times \dots \times U^{(k)}$, we perform the following phase. For every $i \in [k-1]$, and $j'  \in [N] \setminus U^{(i)}$, we insert the  hypercube with the following largest corner to the data structure:
$$(\underbrace{N, \ldots, N}_{2i-2 \text{ coordinates of } N}, j', N-j', N, \ldots,N).$$
Let $Q$ be the union of hypercubes in the current state of the data structure. We query the data structure to get $V_0 = \vol(Q)$. 
Then, for each $j \in [N]$, we insert a hypercube with the largest corner 
$$(N, \ldots, N, j).$$
 After this insertion, we query the data structure for the  volume of the union of the hypercubes, and denote the answer of the query by $V_j$. After the query, we delete the hypercube with the largest corner 
$(N, \ldots, N, j)$ and proceed to the next $j$. After we finish for every $j$, we delete all points added in the current phase and end the phase.

Let $f_Q(j) = \vol\left(Q \cap \left([0, N]^{2k-2} \times [j-1, j]\right)\right)$ for $j \in [N]$, where $Q$ was defined above. We then show  the following claims. 
\begin{claim}
\label{cl:klee1}
 The value of $f_Q(j)$ equals the volume of the set of points in $[0, N]^{2k-2}$ dominated by at least one of the following points in $\R^{2k-2}$:
\[\left\{
  \begin{array}{llr}
    (\underbrace{N, \ldots, N}_{2i-2 \text{ coordinates of } N}, j', N-j', N, \ldots)  &: i \in [k-1], j'  \in [N] \setminus U^{(i)} &\\
    (a_1 - \delta a_k, N-a_1, a_2, N-a_2, \ldots, & a_{k-1}, N - a_{k-1}) & \\
     &: (a_1, \ldots, a_k) \in M \cap \left(U^{(1)} \times \cdots \times U^{(k-1)} \times \{j, \ldots, N\}\right). &
  \end{array}
\right.
\]
\end{claim}
\begin{proof}
Since the largest corners of the hypercubes we add all have integral $(2k-1)$-th coordinate, and all hypercubes have integral side lengths $N$,  the intersections of the last ($(2k-1)$-th) dimension  of these hypercubes with $[j-1, j]$ have lengths either $0$ or $1$. Therefore, it suffices to consider those hypercubes that completely cover $[j-1, j]$ in the last dimension, and their projection onto the first $(2k-2)$ dimensions. Therefore, $f_Q(j)$ equals the volume of the set of points in $[0, N]^{2k-2}$ dominated by one of the following points (which are the projections of all largest corners in $Q$ that lie in $[0, N]^{2k-2} \times [j, N]$), times the length of $[j-1, j]$ (which is $1$):
\[\left\{
  \begin{array}{llr}
    (\underbrace{N, \ldots, N}_{2i-2 \text{ coordinates of } N}, j', N-j', N, \ldots) & : i \in [k-1], j'  \in [N] \setminus U^{(i)} &\\
    (a_1 - \delta a_k, N-a_1, a_2, N-a_2, \ldots, & a_{k-1}, N - a_{k-1}) & \\
    & : (a_1, \ldots, a_k) \in M \cap \left([N]^{k-1} \times \{j, \ldots, N\}\right). &
  \end{array}
\right.
\]
It remains to show that the volume does not change if we exclude $(a_1 - \delta a_k, N-a_1, a_2, N-a_2, \ldots, a_{k-1}, N - a_{k-1})$ from the above list where $(a_1, \ldots, a_k) \in M \cap \left([N]^{k-1} \times \{j, \ldots, N\}\right)$ such that $a_i \not \in U^{(i)}$ for some $i \in [k-1]$. Such a point $p$ is dominated by another point $q = (\underbrace{N, \ldots, N}_{2i-2 \text{ coordinates of } N}, a_i, N-a_i, N, \ldots)$ in the list, so $q$ dominates all points dominated by $p$. Therefore, excluding $p$ does not change the volume. 
\end{proof}

\begin{claim}
\label{cl:klee2}
For $j \in [N-1]$, $V_j - V_{j-1} = V_{j+1}-V_j$ if and only if $M \cap \left(U^{(1)} \times \cdots \times U^{(k-1)} \times \{j\}\right)$ is empty. 
\end{claim}
\begin{proof}
First, for any $j \in [N]$, $V_j = \vol(Q \cup ([0, N]^{2k-2} \times [j-N, j]))$ by definition. Since $Q$ already covers all volume in $[-N, N]^{2k-1}$ except the nonnegative orthant, $V_j$ can be further written as $\vol(Q \cup ([0, N]^{2k-2} \times [0, j]))$. Note that this is also true for $j=0$. 

Now we can write $V_j - V_{j-1}$ as $\vol(Q \cup ([0, N]^{2k-2} \times [0, j])) - \vol(Q \cup ([0, N]^{2k-2} \times [0, j-1]))$ for any $j \in [N]$. Note that $Q \cup ([0, N]^{2k-2} \times [0, j])$ and $Q \cup ([0, N]^{2k-2} \times [0, j-1])$ are identical except in the region $[0, N]^{2k-2} \times [j-1, j]$. Therefore, 
\begin{equation*}
    \begin{split}
        V_j - V_{j-1} =& \vol(Q \cup ([0, N]^{2k-2} \times [0, j])) - \vol(Q \cup ([0, N]^{2k-2} \times [0, j-1]))\\
        =& \vol\left( \left(Q \cup \left([0, N]^{2k-2} \times [0, j]\right)\right) \cap \left([0, N]^{2k-2} \times [j-1, j]\right)\right)\\
        & - \vol\left( \left(Q \cup \left([0, N]^{2k-2} \times [0, j-1]\right)\right) \cap \left([0, N]^{2k-2} \times [j-1, j]\right)\right)\\
        =& N^{2k-2} - f_Q(j).
    \end{split}
\end{equation*}
Thus, we can replace the condition $V_j - V_{j-1} = V_{j+1}-V_j$ with $f_Q(j) = f_Q(j+1)$. 

First, suppose $M \cap \left(U^{(1)} \times \cdots \times U^{(k-1)} \times \{j\}\right)$ is empty. In this case, the list of points in the statement of \cref{cl:klee1} is the same for $f_Q(j)$ and $f_Q(j+1)$, and thus $f_Q(j)=f_Q(j+1)$. 

Conversely, suppose $M \cap \left(U^{(1)} \times \cdots \times U^{(k-1)} \times \{j\}\right)$ is not empty. In this case, the list of points in the statement of \cref{cl:klee1} for $f_Q(j+1)$ is a proper subset of the list of points for $f_Q(j)$. The points that are in $f_Q(j)$'s list while not in $f_Q(j+1)$'s list are 
$$ (a_1 - \delta a_k, N-a_1, a_2, N-a_2, \ldots, a_{k-1}, N - a_{k-1}), $$
where $(a_1, \ldots, a_k) \in M \cap \left(U^{(1)} \times \cdots \times U^{(k-1)} \times \{j\}\right).$ 
It is not difficult to verify that these points are maximal points among the list of points for $f_Q(j)$ (using ideas from the proof of Theorem~\ref{thm:skyline}), so removing any of them strictly decreases the volume of the set of points in $[0, N]^{2k-2}$ dominated by the points in the list. It implies that $f_Q(j+1) < f_Q(j)$. 
\end{proof}

Now let us complete the reduction. For each $j \in [N-1]$, we can use \cref{cl:klee2} to test whether $M \cap \left(U^{(1)} \times \cdots \times U^{(k-1)} \times \{j\}\right)$ is empty in $O(1)$ time. Furthermore, we can also test whether $M \cap \left(U^{(1)} \times \cdots \times U^{(k-1)} \times \{N\}\right)$ is empty via an $O(N^{k-1})$ time brute-force algorithm. Given these results, we can determine if $M \cap \left(U^{(1)} \times \cdots \times U^{(k)}\right)$ is empty in $O(N)$ additional time. 

The pre-processing time of the reduction is $\poly(n) = \poly(N)$. For each OuMv$_k$ query, we call the data structure $O(N)$ times, which cost $O(N \cdot n^{1-1/k-\epsilon}) = O(N^{k-k\epsilon})$ time. We also spend $O(N^{k-1})$ additional time for each query. Therefore, each query takes $O(N^{k-k\epsilon} + N^{k-1})$ time for $\epsilon > 0$, which contradicts the OuMv$_k$ hypothesis. Thus, assuming the OuMv$_k$ hypothesis, there is no data structure for Dynamic Klee's measure for unit hypercubes in $\mathbb{R}^{2k-1}$ with $\poly(n)$ pre-processing time, $O(n^{1-1/k -\epsilon})$ update and query time for $\epsilon > 0$. 

Clearly, the lower bound also works for data structures in the semi-online model, since in fact, we know the deletion time of  all hypercubes when they are inserted. 
\end{proof}

\subsection{Chan's Halfspace Problem}

We finally show a lower bound for Chan's Halfspace Problem, which was considered in \cite{dyn:Chan03}. 

\begin{prob}[Chan's Halfspace Problem]
Fix a constant integer parameter $d \ge 1$, and let $\square$ be a fixed function from multiple numbers to one number that is decomposable and allows computing $\square (S+j)$ from $\square S$ in constant time. Maintain  a dynamic set $H$ of hyperplanes in $\R^d$, and  a dynamic set $Q$  of points in $\R^d$. Each update operation can insert (resp.\ delete) a hyperplane to (resp.\ from) $H$ or a point to (resp.\ from) $Q$. Each query asks to compute $\square c_H(Q) = \square \{c_H(q): q \in Q\}$, where $c_H(q)$ is the number of hyperplanes in $H$ that contains $q$. 
\end{prob}

\chan*
\begin{proof}
Assume for the sake of contradiction that such an efficient data structure exists. 
We will reduce from an OuMv$_k$ instance of dimension $N = \Theta(n^{1/k})$. Let $M \subseteq [N]^k$ be the input set in OuMv$_k$.

For every tuple $(a_1, \ldots, a_k) \in M$, we add a  point $(a_1,\dots, a_k)\in \R^k$ into  point set $Q$.  Then we use the pre-processing part of the assumed data structure to pre-process these points in $\poly(n)$ time. 

For every OuMv$_k$ query $U^{(1)} \times \dots \times U^{(k)}$, we perform the following phase. For every $i \in [k]$, and $j\in U^{(i)}$, we insert the following two halfspaces into $H$:
\begin{align*}
 \{ (x_1,\dots,x_k) \in \R^k  : x_i < j-0.5 \},\\
 \{ (x_1,\dots,x_k) \in \R^k  : x_i > j+0.5 \}.
\end{align*}
Observe that, considering the $2|U^{(i)}|$ halfspaces inserted for each $i$, point $(a_1,\dots,a_k)$ is contained in $|U^{(i)}|$ of them if $a_i\notin U^{(i)}$, or is contained in $|U^{(i)}|-1$ of them if $a_i\in U^{(i)}$. 
Then, we use the query operation to obtain $\min c_H(Q)$, which equals $\sum_{i=1}^k (|U^{(i)}|-1)$ if and only if there exists $(a_1,\dots,a_k) \in M \cap
U^{(1)} \times \dots \times U^{(k)}$, i.e., the answer to this OuMv$_k$ query is YES.  At the end of this phase, we remove the added halfspaces from $H$.

The pre-processing time of the reduction is $\poly(n) = \poly(N)$. For each OuMv$_k$ query, we call the data structure $O(N)$ times, which cost $O(N \cdot n^{1-1/k-\epsilon}) = O(N^{k-k\epsilon})$ time.  Therefore, each query takes $O(N^{k-k\epsilon})$ time for $\epsilon > 0$, which contradicts  the OuMv$_k$ hypothesis. Thus, assuming the OuMv$_k$ hypothesis, there is no data structure for Chan's Halfspace problem in $\mathbb{R}^{k}$ with $\poly(n)$ pre-processing time, $O(n^{1-1/k -\epsilon})$ update and query time for $\epsilon > 0$. 
\end{proof}

%% file: generalization_sec_overview.tex
In this section, we show hardness for generalizations of problems that were known to be OMv-hard \cite{HenzingerKNS15}, including generalizations of $s$-Triangle Detection, Erickson’s problem and Langerman’s problem. 

%% file: Erickson_etc.tex
\subsection{Dynamic \texorpdfstring{$s$}{s}-\texorpdfstring{$k$}{k}-Uniform \texorpdfstring{$(k+1)$}{(k+1)}-Hyperclique}

\begin{prob}[Dynamic $s$-$k$-Uniform $(k+1)$-Hyperclique]
Maintain an $n$-node $k$-uniform hypergraph with a fixed node $s$ and support inserting a hyperedge, deleting a hyperedge, and querying whether $s$ is in a $k$-uniform $(k+1)$-hyperclique. 
\end{prob}

There are two naive algorithms for this problem. The first algorithm has $\tO(1)$ update time and $\tO(n^k)$ query time: during an update, it does not do any real work besides recording the update; during a query, it enumerates all tuples of $k$ vertices in the graph and check whether they form a $k$-uniform $(k+1)$-hyperclique with $s$. The second algorithm has $\tO(n)$ update time and $\tO(1)$ query time: it maintains the number of $k$-uniform $(k+1)$-hypercliques each vertex is in; for each update inserting or deleting a $k$-uniform hyperedge,  it enumerates all the $O(n)$ tuples of $(k+1)$ of vertices that contain this hyperedge and updates the number of $k$-uniform $(k+1)$-hypercliques each vertex is in accordingly; during a query, the algorithm only needs to read the count for $s$, which takes $\tO(1)$ time. 

We show that these two simple algorithms are actually optimal under the OuMv$_k$ hypothesis, generalizing the OMv-hardness of Dynamic $s$-Triangle Detection \cite{HenzingerKNS15}. 

\begin{theorem}
\label{thm:s-k-hyperclique}
Let $k \ge 2$ be a positive integer. Assuming the OuMv$_k$ hypothesis, there is no data structure for Dynamic $s$-$k$-Uniform $(k+1)$-Hyperclique with $\poly(n)$ pre-processing time, $O(n^{1-\eps})$ amortized update time and $O(n^{k-\eps})$ amortized query time for $\eps > 0$. 
\end{theorem}
\begin{proof}
Assume for the sake of contradiction that such an efficient data structure exists. We will reduce from an OuMv$_k$ instance over $[N]^k$ for $N = \Theta(n)$. Let $M \subseteq [N]^k$ be the input set of OuMv$_k$. 

We construct a Dynamic $s$-$k$-Uniform $(k+1)$-Hyperclique instance on vertex set $\{s\} \cup [N] \times [k]$. For every tuple $(a_1, \ldots, a_k) \in M$, we add the following hyperedge to the initial graph:
$$\left\{(a_1, 1), \ldots, (a_k, k) \right\}.$$
Then we let the assumed data structure to pre-process this hypergraph in $\poly(N)$ time. 

For every OuMv$_k$ query $U^{(1)} \times \dots \times U^{(k)}$, we perform the following phase. For every size $k-1$ subset of $[k]$ consisting of elements $t_1, \ldots, t_{k-1}$, we enumerate all tuples $(a_{t_1}, \ldots, a_{t_{k-1}}) \in U^{(t_1)} \times \cdots \times U^{(t_{k-1})}$, and insert the following hyperedge to the hypergraph via the data structure:
$$\left\{s, (a_{t_1}, t_1), \ldots, (a_{t_{k-1}}, t_{k-1})\right\}.$$
After we insert all the edges, we query whether $s$ is in a $k$-uniform $(k+1)$-hyperclique. We claim that $s$ is in a $k$-uniform $(k+1)$-hyperclique if and only if the answer to the OuMv$_k$ query is YES. First, suppose $(a_1, \ldots, a_k) \in \left( U^{(1)} \times \ldots \times U^{(k)}\right) \cap M$, then it is easy to check the vertices $\left\{s, (a_1, 1), \ldots, (a_k, k)\right\}$ form a $k$-uniform $(k+1)$-hyperclique. Now we consider the converse direction. Since the graph is a $(k+1)$-partite  $k$-uniform hypergraph by construction, with $s$ being on its own part and vertices $[N] \times \{i\}$ form a part for each $i \in [k]$, any $(k+1)$-hyperclique must use one vertex from each part and thus have the form $\left\{s, (a_1, 1), \ldots, (a_k, k)\right\}$. Therefore, $(a_1, \ldots, a_k)$ is in both $U^{(1)} \times \dots \times U^{(k)}$ and $M$, and thus the answer to the OuMv$_k$ query is YES. 

At the end of each phase, we remove all hyperedges inserted in this phase. 

The total number of data structure updates we make for each OuMv$_k$ query is $O(n^{k-1})$, and the total number of data structure queries is $O(1)$. Therefore, the running time of each phase is 
$O(n^{k-1} \cdot n^{1-\eps} + n^{k-\eps}) = O(n^{k-\eps})$, which contradicts the OuMv$_k$ hypothesis. Therefore, assuming the OuMv$_k$ hypothesis, there is no data structure for Dynamic $s$-$k$-Uniform $(k+1)$-Hyperclique with $\poly(n)$ pre-processing time, $O(n^{1-\eps})$ update time and $O(n^{k-\eps})$ query time for $\eps > 0$. 
\end{proof}

\subsection{\texorpdfstring{$k$}{k}-Dimensional Erickson’s Problem}

We first define the following generalization of Erickson’s problem. The original Erickson’s problem \cite{patrascu2010towards} corresponds to $2$-Dimensional Erickson’s problem. 

\begin{prob}[$k$-Dimensional Erickson’s Problem]
Maintain a $k$-dimensional tensor on integers of size $n \times \cdots \times n$ and support incrementing all entries whose $i$-th coordinate is $x$ for some $i$ and $x$ and querying the maximum value in the tensor. 
\end{prob}

There are two possible brute-force algorithms. The first one maintains a set of all the entries in the tensor, so that it runs in $\tO(n^{k-1})$ time per update to update all the changed entries in the tensor
and $\tO(1)$ time per query to extract the maximum value from the set. The second brute-force algorithm maintains the amount of increments we perform for each $i$ and $x$, so that given an entry, we can compute its value in $\tO(k) = \tO(1)$ time since $k$ is a constant. Thus, this algorithm runs in $\tO(1)$ time per update and $\tO(n^k)$ time per query. 

We show that these two algorithms are essentially optimal. 

\begin{theorem}
\label{thm:erickson}
Let $k \ge 2$ be a positive integer. Assuming the OuMv$_k$ hypothesis, there is no data structure for $k$-Dimensional Erickson’s problem with $\poly(n)$ pre-processing time, $O(n^{k-1-\eps})$ amortized update time and $O(n^{k-\eps})$ amortized query time for $\eps > 0$. 
\end{theorem}
\begin{proof}
Assume for the sake of contradiction that such an efficient data structure exists. We will reduce from an OuMv$_k$ instance over $[n]^k$. Let $M \subseteq [n]^k$ be the input set of OuMv$_k$. 

First, we create an $n \times \cdots \times n$ tensor $T$ where all entries are $0$. Then for every $(a_1, \ldots, a_k) \in M$, we set $T[(a_1, \ldots, a_k)]$ to be $1$. We then use the assumed data structure for $k$-Dimensional Erickson’s problem to pre-process $T$ in $\poly(n)$ time. 

For every OuMv$_k$ query $U^{(1)} \times \dots \times U^{(k)}$, we perform the following phase. For each $i \in [k]$, and each $x \in U^{(i)}$,  we increment all entries in $T$ whose $i$-th coordinate is $x$. Then we query the maximum value in the tensor. After the query, for each $i \in [k]$, and each $x \not \in U^{(i)}$,  we increment all entries in $T$ whose $i$-th coordinate is $x$. 

Clearly, the updates in each phase increment all entries in $T$ a number of  $k$ times. During the $f$-th phase, we increment all entries in $T$ whose $i$-th coordinate is $x$ for $i \in [k]$ and $x \in  U^{(i)}$ before we ask the query from the data structure. Therefore, at the time of that query, an entry $(a_1, \ldots, a_k)$ of the tensor has value $k+1+(f-1)k$ if and only $(a_1, \ldots, a_k) \in M$ and $a_i \in U^{(i)}$ for every $i \in [k]$. Also, $k+1+(f-1)k$ is clearly  an upper bound for all values in the tensor. Therefore, the maximum value returned by the query is $k+1+(f-1)k$ if and only if the answer to the OuMv$_k$ is YES. 

The total number of data structure updates we make for each OuMv$_k$ query is $O(n)$, and the total number of data structure queries is $O(1)$. Therefore, the running time of each phase is $O(n \cdot n^{k-1-\eps} + n^{k-\eps}) = O(n^{k-\eps})$, contradicting  the OuMv$_k$ hypothesis. Therefore, assuming the OuMv$_k$ hypothesis, there is no data structure for $k$-Dimensional Erickson’s problem with $\poly(n)$ pre-processing time, $O(n^{k-1-\eps})$ update time and $O(n^{k-\eps})$ query time for $\eps > 0$. 
\end{proof}

\subsection{\texorpdfstring{$d$}{d}-Dimensional Langerman’s Problem}

We define the high-dimensional variant of the Langerman’s problem. The original Langerman’s problem~\cite{patrascu2010towards} corresponds to $d=1$. 

\begin{prob}[$d$-Dimensional Langerman’s problem]
Maintain a $d$-dimensional tensor $T$ on integers of size $n \times \cdots \times n$ and support updating the value of an entry. For each query, determine whether there exists $x \in [n]^d$, such that $P[x] \stackrel{\text{def}}{=} \sum_{y \text{ dominated by } x} T[y] = 0$. 
\end{prob}

\begin{prop}
\label{prop:Langerman_algo}
There exists a data structure for the $d$-Dimensional Langerman’s problem with polynomial pre-processing time and $\tO(n^{d^2/(d+1)})$ update and query time. 
\end{prop}
\begin{proof}

Let $B = \Theta(n^{1/(d+1)})$ be a parameter. Without loss of generality, we assume $n$ is a multiple of $B$. 
For any $d$-dimensional vector $x$, we use $\lfloor x / B \rfloor$ to denote the vector $(\lfloor x_1 / B \rfloor, \ldots, \lfloor x_d / B\rfloor)$. 
We split $T$ to $O((n/B)^{d})$ pieces of sub-tensors of sizes $B \times \cdots \times B$, so that for any two entries $x, x'$ in the same sub-tensor, $\lfloor x/B\rfloor = \lfloor x'/B \rfloor$. 

We also maintain the following sub-data structures:
\begin{enumerate}
    \item For every $y \in [n/B]^d$, maintain $P[By]$.
    \item For every $x \in [n]^d$, maintain $P[x] - P[B \lfloor x / B \rfloor]$. 
    \item For every piece of sub-tensor $T'$ of size $B \times \cdots \times B$, maintain a multi-set containing all values of $P[x] - P[B \lfloor x / B \rfloor]$ for $x \in T'$. 
\end{enumerate}
We can clearly initialize these sub-data structures in $\tO(n^d)$ time. 

For each update changing $T[z]$ from value $p$ to value $q$, we can update these sub-data structures as follows:
\begin{enumerate}
    \item For every $y \in [n/B]^d$ such that $By$ dominates $z$, we add $q-p$ to $P[By]$. This step takes $\tO(n^{d^2/(d+1)})$ time. 
    \item For some $x$, the value $P[x] - P[B \lfloor x / B \rfloor]$ will be affected if and only if $z$ is dominated by $x$ while $z$ is not dominated by $B \lfloor x / B \rfloor$. If this happens, then at least one coordinate of $x$ differs by at most $B$ from $z$. Thus, for every $x$ that has at least one coordinate differing by at most $B$ from that coordinate of $z$, we check whether $P[x] - P[B \lfloor x / B \rfloor]$ will be affected by the change of $T[z]$, and add $q-p$ to $P[x] - P[B \lfloor x / B \rfloor]$ if the check passes. The number of such $x$ is at most $O(n^{d-1}B) = O(n^{d^2/(d+1)})$. Therefore, this step takes $\tO(n^{d^2/(d+1)})$ time. 
    \item This step is relatively easy given the results from the previous step, and will take $\tO(n^{d^2/(d+1)})$ time. 
\end{enumerate}

Now we discuss how the data structure handles queries. Consider each sub-tensor $T'$. For any $x \in T'$, $P[x] = \left(P[x] - P[B \lfloor x / B \rfloor]\right) + P[B \lfloor x / B \rfloor]$. Also, the value $P[B \lfloor x / B \rfloor]$ is the same for every $x \in T'$ by the definition of a sub-tensor. Denote this value by $V$. Therefore, we essentially need to determine whether there exists $x \in T'$ such that $P[x] - P[B \lfloor x / B \rfloor] = -V$. This can be answered in $\tO(1)$ time per sub-tensor using the third sub-data structure. Since there are $O((n/B)^d) = O(n^{d^2/(d+1)})$ sub-tensors in total, each query takes $\tO(n^{d^2/(d+1)})$ time. 
\end{proof}

We show that the data structure in \cref{prop:Langerman_algo} is nearly-optimal under the OuMv$_k$ hypothesis. 

\begin{theorem}
\label{thm:langerman}
Let $k \ge 2$ be a positive integer. Assuming the OuMv$_k$ hypothesis, there is no data structure for $d$-Dimensional Langerman’s problem  with $\poly(n)$ pre-processing time, $O(n^{d^2/(d+1)-\eps})$ amortized update  and  query time for $\eps > 0$ where $d = k - 1$. 
\end{theorem}
\begin{proof}
Assume for the sake of contradiction that such an efficient data structure exists. We will reduce from an OuMv$_k$ instance over $[N]^k$. 
Let $M \subseteq [N]^k$ be the input set of OuMv$_k$. 

Let $B = N^{1/d}$ be a parameter (recall $d = k - 1$). Without loss of generality, assume $B$ is an integer. Let $f: [B]^d \rightarrow [N]$ be an arbitrary bijection between $[B]^d$ and $[N]$. 
For every $a=(a_1, \ldots, a_d) \in [N]^d$, we create the following tensor $S_{a}$ of dimension $(B+1)^d$.  For any $(a_1, \ldots, a_{d+1}) \in M$, we set $S_{a}[f^{-1}(a_{d+1})]$ to be $a_{d+1}$. All other entries of $S_{a}$ are zeros. We also create the following tensor $A_{a}$ of dimension $(B+1)^d$, where 
$$A_{a}[y] = \sum_{b \in \{0, 1\}^d} (-1)^{\sum_i b_i} S_{a}[y - b], $$
for every $y \in [B+1]^d$, where $S_{a}[y - b]$ is regarded as $0$ if $y - b$ is out of bound. We can view $S_{a}$ as the tensor that corresponds to the $d$-dimensional prefix sums of $A_{a}[y]$. 

Define $n = (B+1)N$. Then we create the following tensor $T$ of dimensions $n^d  = \left( (B+1) N \right)^d$ for the $d$-Dimensional Langerman’s problem. Informally, this tensor can be viewed as a tensor of dimension $N^d$, with its $(a_1, \ldots, a_d)$-th entry replaced by $A_{a_1, \ldots, a_d}$. More formally, for any $a \in [N]^d$ and any $y \in [B+1]^d$, we set $T[(B+1)(a-\vec{\mathbf{1}}) + y] = A_a[y]$, where $\vec{\mathbf{1}}$ is the all-ones vector. 

Recall for any $x \in [n]^d$, we define $P[x]$ to be $\sum_{x' \le x} T[x']$ where $x' \le x$ if and only if $x'_i \le x_i$ for every $i \in [d]$. This tensor $T$ has the following nice property: for any $a \in [N]^d$ and any $y \in [B+1]^d$, 
\begin{equation*}
    \begin{split}
        P[(B+1)(a-\vec{\mathbf{1}}) + y] &= \sum_{a' \in [N]^d} \left( \sum_{\substack{y' \in [B+1]^d \\ (B+1)(a'-\vec{\mathbf{1}}) + y' \le (B+1)(a-\vec{\mathbf{1}})+y}} T[(B+1)(a'-\vec{\mathbf{1}}) + y']\right)\\
        &= \sum_{\substack{a' \in [N]^d \\ a' \le a}} \left( \sum_{\substack{y' \in [B+1]^d \\ \forall i, y'_i \le y_i \text{ if } a'_i = a_i }} A_{a'}[y']\right)\\
        &= \sum_{\substack{a' \in [N]^d \\ a' \le a}} S_{a'}[g(a, a', y)],
    \end{split}
\end{equation*}
where $g(a, a', y)_i = y_i$ if $a_i = a_i'$ and $g(a, a', y)_i = B+1$ otherwise. Furthermore, by the definition of $S_{a'}$, if any coordinate of $g(a, a', y)$ is $B+1$, $S_{a'}[g(a, a', y)] = 0$. Therefore, the above formula can be further simplified to $S_a[y]$, which is $f(y)$ if  $y \in [B]^d$ and $(a_1, \ldots, a_d, f(y)) \in M$ and $0$ otherwise. 

We then feed the tensor $T$ to the pre-processing phase of the assumed data structure, which takes $\poly(n)$ time.

For every OuMv$_k$ query $U^{(1)} \times \dots \times U^{(d+1)}$, we perform the following phase. For every $i \in [d]$, and every $j \not \in U^{(i)}$, we add $1000N$ to the entry $\left(1, \ldots, 1, \underbrace{(j-1)(B+1)+1}_{i\text{th coordinate}} , 1, \ldots, 1 \right)$, and add $-1000N$ to the entry  $\left(1, \ldots, 1, \underbrace{j(B+1)}_{i\text{th coordinate}} , 1, \ldots, 1 \right)$. Then for every $j \in U^{(d+1)}$, we add $-j$ to $T[(1, \ldots, 1)]$ and immediately perform a query in the data structure. We claim that at this point, $T$ has a zero prefix sum if and only if there exists $(a_1, \ldots, a_d) \in U^{(1)} \times \dots \times U^{(d)}$ such that $(a_1, \ldots, a_d, j) \in M$. 

First, suppose there exists $(a_1, \ldots, a_d) \in U^{(1)} \times \dots \times U^{(d)}$ such that $(a_1, \ldots, a_d, j) \in M$. Since $P[(B+1)(a-\vec{\mathbf{1}}) + f^{-1}(j)]$ has value $j$ before the phase, and the changes made in this phase increase its value by $-j$,  $T$ has a zero prefix sum. Conversely, suppose there exist $a \in [N]^d$ and $y \in [B+1]^d$ such that $P[(B+1)(a-\vec{\mathbf{1}})+y] = 0$. For any $i \in [d]$, $a_i$ must be contained in $U^{(i)}$, since otherwise, $1000N$ is added to $P[(B+1)(a-\vec{\mathbf{1}})+y]$, making it impossible to be zero. Then the current value of $P[(B+1)(a-\vec{\mathbf{1}})+y]$ equals its value before the phase plus $-j$. In order for its current value to be $0$, its value before the phase must be $j$, which implies $f(y) = j$ and $(a_1, \ldots, a_d, j) \in M$. 

After the query, we add $j$ back to $T[(1, \ldots, 1)]$ and proceed to the next $j$. Before the end of the phase, we revert all changes we make during the phase. 

The total number of data structure updates and queries are $\Theta(N)$. Therefore, the running time of each phase is $O(N \cdot n^{d^2/(d+1)-\eps}) = O(N^{d+1- (d+1)\eps / d})$, contradicting the OuMv$_k$ hypothesis. Therefore, assuming the OuMv$_k$ hypothesis, there is no data structure for $d$-Dimensional Langerman’s problem  with $\poly(n)$ pre-processing time, $O(n^{d^2/(d+1)-\eps})$ update  and  query time for $\eps > 0$ where $d = k - 1$. 
\end{proof}

%% file: appendix.tex
\begin{prop}
\cref{hypo:oumvk1} and \cref{hypo:oumvk2} are equivalent. 
\end{prop}
\begin{proof}
Clearly, \cref{hypo:oumvk2} implies \cref{hypo:oumvk1}. Thus, it suffices to show \cref{hypo:oumvk1} implies \cref{hypo:oumvk2}. We will prove the contrapositive.

Suppose there is an algorithm for the OuMv$_k$ problem with $O(n^t)$ pre-processing time and $O(n^{\gamma + k-\eps})$ total query time for $n^\gamma$ queries, for some $t, \gamma, \eps>0$. We aim to show an algorithm for OuMv$_k$ on input $M \subseteq [N]^k$ and $N$ queries that run in $O(N^{1+k - \eps'})$ time.

Let $\delta \in (0, 1)$ be a small enough constant to be fixed later. We first split $[N]^k$ to sub-tensors of dimensions $N^\delta \times \cdots \times N^\delta$. Let $n = N^\delta$. For each of the $N^{k(1-\delta)}$ sub-tensors, we run the assumed $O(n^t)$ time pre-processing algorithm for OuMv$_k$. Overall, this step takes $O(N^{k+(t-k)\delta})$ time. 

Then in each phase, we handle $n^\gamma$ queries. Here, we need $n^\gamma = O(N)$, which is equivalent to $\delta \gamma \le 1$. For each of the $n^\gamma$ queries, we run the assumed algorithm for OuMv$_k$ on every sub-tensor with its corresponding portion in the query. Over all the $n^\gamma$ queries, the running time is $O(n^{\gamma + k - \eps} \cdot N^{k(1-\delta)})$. After each $n^\gamma$ queries, we recover the data structure to its state at the beginning of the phase, so that the data structure will be ready for the next $n^\gamma$ queries. The recovering time is also $O(n^{\gamma + k - \eps} \cdot N^{k(1-\delta)})$. Therefore, the overall time for handling $N$ queries is $O(N^{1+k-\delta \eps})$. 

Therefore, in order for the total time of our algorithm to be $O(N^{1+k-\eps'})$ for some $\eps' > 0$, it suffices to take any $\delta \in (0, 1)$ such that $(t-k) \delta < 1$ and $\delta \gamma \le 1$, which is clearly possible.  
\end{proof}